\documentclass[11pt]{article}

\usepackage{graphicx}
\usepackage{float}
\usepackage{lmodern}
\usepackage{booktabs}
\usepackage{amsmath,amsfonts,bbm,bm,mathtools,relsize,exscale}
\usepackage{dsfont}
\usepackage[OT1]{fontenc}
\usepackage[utf8]{inputenc}
\usepackage[inline,shortlabels]{enumitem}
\graphicspath{ {figs/} }
\usepackage{listings}
\usepackage{upquote}
\usepackage{cancel}
\usepackage{algorithm}
\usepackage{algpseudocode}
\makeatletter
\newcommand{\AlgStep}[1]{\State\parbox[t]{\dimexpr\linewidth-\ALG@thistlm\relax}{%
  \raggedright #1\strut}}
\makeatother

\usepackage{refcount,nameref,zref-xr,zref-user} % nameref before zref-xr
\zxrsetup{toltxlabel} % allows use of the \ref command as usual
\usepackage{url}

\let\href\undefined
\makeatletter
\@ifpackageloaded{dgruyter}{
  \usepackage{amsthm}
  \usepackage[round]{natbib}

  \newtheorem{theorem}{Theorem}
  \newtheorem{lemma}{Lemma}
  \newtheorem{corollary}{Corollary}
  \newtheorem{remark}{Remark}
  \newtheorem{definition}{Definition}
  {\theoremstyle{definition}\newtheorem{assumption}{}}
  {\theoremstyle{definition}\newtheorem{condition}{}}
  {\theoremstyle{definition}\newtheorem{example}{Example}[section]}
}{%
  \usepackage{arxiv}
  \usepackage{standalone}
  \usepackage{multirow}
  \usepackage{setspace}
  \usepackage{amsthm}
  \usepackage[round]{natbib}

  \newtheorem{theorem}{Theorem}
  \newtheorem{lemma}{Lemma}
  \newtheorem{corollary}{Corollary}
  \newtheorem{remark}{Remark}
  
  {\theoremstyle{definition}}
  {\theoremstyle{definition}}
  {\theoremstyle{definition}}

  \newcommand{\authorlist}{
    Salvador V. Balkus \\
    Department of Biostatistics,\\
    Harvard Chan School of Public Health\\
    \texttt{sbalkus@g.harvard.edu}\\
    \And
    Christian Testa \\
    Department of Biostatistics,\\
    Harvard Chan School of Public Health\\
    \texttt{ctesta@hsph.harvard.edu}\\
    \And
    Nima S.~Hejazi \\
    Department of Biostatistics,\\
    Harvard Chan School of Public Health\\
    \texttt{nhejazi@hsph.harvard.edu}\\
  }
}
\makeatother

\newcommand{\titlepaper}{A Riesz Representer Perspective on Targeted Learning}

\newcommand{\abstractpaper}{
As research in causal inference has sought to address more complex scientific
questions, the number of specialized causal estimands in the field has proliferated.
Recognition that many of these estimands share a common linear form has generated
interest in simplifying estimation procedures using Riesz representers. In this
work, we construct a targeted minimum loss-based estimation procedure for
nested linear functionals, leveraging Riesz representers of a general recursive
form. The proposed method unifies asymptotically efficient estimation for a
variety of statistical estimands that originate in causal inference, including
the effects of time-varying treatments under treatment-confounder feedback and
direct and indirect effects from causal mediation analysis. We demonstrate how
this reduces the need for laborious and technically challenging
mathematical derivations when constructing estimators of common statistical
estimands under complex forms of censoring and sampling. We investigate and
validate the properties of the proposed procedures in numerical experiments,
discuss open-source software facilitating their implementation, and illustrate
their application in a re-analysis of data from an HIV vaccine efficacy trial.
}

\AtEndDocument{\refstepcounter{theorem}\label{finalthm}}
\AtEndDocument{\refstepcounter{equation}\label{finaleq}}
\AtEndDocument{\refstepcounter{lemma}\label{finallemma}}

\newcommand{\E}{\mathsf{E}}
\renewcommand{\P}{\mathsf{P}}

\newcommand{\M}{\mathcal{M}}
\newcommand{\R}{\mathbb{R}}

\renewcommand{\Pr}{\mathbb{P}}
\renewcommand{\P}{\mathsf{P}}
\newcommand{\1}{\mathbbm{1}}

\newcommand{\link}{\text{link}}

\usepackage[dvipsnames]{xcolor}

\newcommand{\ijb}[1]{#1}
\zexternaldocument*[supp:]{sm}

\makeatletter
\@ifpackageloaded{dgruyter}{}{
  \author{\authorlist}
  \date{\today}
  \title{\titlepaper}
}
\makeatother

\begin{document}

\makeatletter
\@ifpackageloaded{dgruyter}{
  % \articletype{...}
  \author[1]{Salvador V.~Balkus}
  \author[2]{Christian Testa}
  \author*[3]{Nima S.~Hejazi}
  \runningauthor{S.V.~Balkus, C.~Testa, and N.S.~Hejazi}
  \affil[1]{Department of Biostatistics,
    Harvard T.H.~Chan School of Public Health;
    e-mail: \texttt{sbalkus@g.harvard.edu}}
  \affil[2]{Department of Biostatistics,
    Harvard T.H.~Chan School of Public Health;
    e-mail: \texttt{ctesta@hsph.harvard.edu}}
  \affil[3]{Department of Biostatistics,
    Harvard T.H.~Chan School of Public Health;
    e-mail: \texttt{nhejazi@hsph.harvard.edu}}
  \title{\titlepaper}
  \runningtitle{\titlepaper}
  % \subtitle{...}
  \abstract{\abstractpaper}
  % \classification[PACS]{...}
  % \communicated{...}
  % \dedication{...}
  \received{...}
  \accepted{...}
  \journalname{...}
  \journalyear{...}
  \journalvolume{...}
  \journalissue{...}
  \startpage{1}
  \aop
  \DOI{...}
}{}

\maketitle

\@ifpackageloaded{dgruyter}{}{
\begin{abstract}
    \abstractpaper
\end{abstract}
}
\makeatother

%%%%%%%%%%%%%%%%%%%%%%%%%%%%%%%%%%%%%%%%%%%%%%%%%%%%%%%%%%%%%%%%%%%%%%%%%%%%%%%
\section{Introduction}\label{intro}

Contemporaneous advances at the intersection of causal inference and machine
learning in recent decades have led to a rich and thriving research interface
between the two fields. Significant effort in causal inference has been devoted
to developing interpretable, scientifically grounded estimands and outlining
the identification assumptions necessary to re-express counterfactual queries
in terms of quantities estimable from observed data~\citep{pearl2000causality,
van2011targeted, hernan2024}. Simultaneously, the library of machine learning
algorithms for flexible regression has grown
considerably~\citep{friedman2001elements}.

Despite these dual advances, estimation and uncertainty quantification remained
siloed, for a time, between causal inference and machine learning. The use of
semi-parametric theory to bridge these areas~\citep{bickel1993efficient,
vdl2003unified}---in particular the development of theory and methods to
correct asymptotic bias arising from the use of flexible regression procedures
for nuisance function estimation---helped give rise to formal frameworks for
de-biased estimation. Among these frameworks, \textit{targeted
learning}~\citep{van2006targeted, van2011targeted} has outlined a general,
structured template for constructing asymptotically efficient substitution
estimators robust to inconsistent nuisance estimation, and, as such, has helped
to drive an explosion of interest in this interface, more recently termed
\textit{causal machine learning}.

Causal machine learning seeks to unify core principles of causal inference with
advances in machine learning, using the former to derive target statistical
estimands and the latter to facilitate flexible estimation of relevant
components of the data-generating process (i.e., nuisance functions) relevant
to estimation of the target estimand. This use of flexible learning algorithms
for nuisance estimation can help to curb an often severe bias---from model
misspecification---in the final estimator, yet it comes at a cost: the
resulting estimator incurs an asymptotic bias stemming from a mismatch between
the typically slower convergence rates of machine learning algorithms and the
comparatively faster rate attained by more commonly used parametric regression
strategies.

To fix ideas, consider, for example, the statistical functional $\Psi(\P_0) =
\E_{\P_0}[\E_{\P_0}(Y \mid A = a, L)]$, which, under well-studied
identification assumptions, among them consistency, positivity and conditional
exchangeability~\citep{hernan2024}, coincides with the counterfactual mean
$\E[Y^a]$ of an outcome $Y$ where all study units receive a specific level $a$
of treatment $A$. $Y^a$ is the potential
outcome~\citep{rubin1974estimating, neyman1990application} of $Y$ under
treatment level $a \in \mathcal{A}$, and we use $\P_0$ to denote the
data-generating law of the data unit $O$, where we assume that $\P_0 \in \M$,
for $\M$ an unrestricted statistical model (that is, a collection of candidate
data-generating laws $\P$). The statistical target parameter is a mapping
$\Psi: \M \to \R^d$. We let $\overline{Q}_{\P_0}(A, L) \coloneqq \E_{\P_0}(Y
\mid A, L)$ denote the mean of the outcome $Y$ conditional on treatment $A$ and
putative confounders $L$. Going forward, we abbreviate the nuisance function
evaluated at the true distribution $\overline{Q}_{\P_0}$ as $\overline{Q}_0$,
and the nuisance at an arbitrary distribution $\overline{Q}_{\P}$ as
$\overline{Q}$. We note that $\psi_0= \Psi(\overline{Q}_{\P_0}; \P_0) =
\E_{\P_0}[\overline{Q}_{\P_0}(a, L)]$, clarifying that the parameter mapping
$\Psi(\cdot)$ depends on only a specific component $\overline{Q}_{\P_0}$ of the
data-generating law $\P_0$ \ijb{as well as an outer expectation over the
distribution of $L$}.

Given a sample of $n$ independent data units, $O_1, \ldots, O_n$, construction
of an estimator $\psi_n$ of $\psi_0 = \Psi(\P_0)$ reduces to construction of an
estimator $\overline{Q}_{n}(A, L) = \hat{\E}_{\P}(Y \mid A, L)$ of
$\overline{Q}_0$, the nuisance function appearing in the parameter mapping.
Learning $\overline{Q}_0$ can be accomplished by a variety of means, among them
parametric regression~\citep{mccullagh1989generalized}; semi-parametric
regression, e.g., generalized additive models~\citep{hastie1987generalized,
hastie1990generalized}; and machine learning, e.g., decision trees and
forests~\citep{breiman1984classification, breiman2001random}, gradient boosting
machines~\citep{friedman2001greedy}, neural networks~\citep{ripley1996pattern,
goodfellow2016deep}, and stacked regression or ensemble
models~\citep{wolpert1992stacked, breiman1996stacked, vdl2007super}. How the
estimator $\overline{Q}_{n}$ of $\overline{Q}_0$ is constructed impacts the
downstream properties of the estimator $\psi_n = n^{-1} \sum_{i=1}^n
\overline{Q}_{n}(a, L)$; moreover, approaches for constructing the nuisance
estimator $\overline{Q}_{n}$---which strive to obtain optimal estimators of the
nuisance function $\overline{Q}_0$, e.g., by appealing to the asymptotic
optimality of cross-validated loss-based estimation~\citep{vdl2004asymptotic,
dudoit2005asymptotics}---are not generally optimally suited for use in
constructing the estimator $\psi_n$. In particular, the use of data-adaptive
regression strategies necessitates the negotiation of a
misspecification--convergence tradeoff: the estimator $\overline{Q}_{n}$ is
much less likely to be misspecified, but it will attain a slower convergence
rate, leaving the estimator $\psi_n$ with a bias that will not vanish
asymptotically.

To enable the use of machine learning in such statistical estimation problems,
a great deal of research in recent years has centered on developing techniques
to remove such asymptotic bias. Targeted maximum likelihood estimation (TMLE),
or targeted minimum loss-based estimation, introduced by
\cite{van2006targeted}, is a general template for ``de-biasing''' estimators
through an updating procedure generically described as ``targeting,'' which
recognizes that the bias incurred is proportional to how well (or poorly) an
estimating equation based on the efficient influence function (EIF), a key
object in semi-parametric efficiency theory~\citep{bickel1993efficient}, is
solved. By fluctuating the nuisance estimator $\overline{Q}_{n}$ in a manner
that solves the EIF estimating equation, a targeting procedure produces an
updated estimator $\overline{Q}_{n}^{\star}$ that results in a consistent and,
under standard regularity conditions, asymptotically efficient estimator
$\psi_n^{\star} = n^{-1} \sum_{i=1}^n \overline{Q}_{n}^{\star}(a, L)$ of
$\psi_0$; the resulting substitution estimator is also, under certain
regularity conditions~\citep{van2011targeted}, asymptotically efficient,
achieving the lowest possible variance among the class of regular
asymptotically linear (RAL) estimators of the target estimand $\psi_0$. This
latter property is inherited from the EIF, which characterizes the
semi-parametric efficiency bound in large, non-parametric statistical models;
recent reviews of the relevant aspects of semi-parametric
theory~\citep{kennedy2016semiparametric, hines2022demystifying,
kennedy2022semiparametric} provide detailed discussions.

Broadly, derivation of a TMLE algorithm consists of three steps: (1) obtain the
unique EIF in the non-parametric model $\M$ for RAL estimators of the target
estimand; (2) set up a maximum likelihood (or loss-based minimization) problem
with parameter(s) $\varepsilon$, such that the derivative of the loss with
respect to $\varepsilon$ equals the EIF; and (3) update the initial estimator
$\overline{Q}_{n}$ of the nuisance function $\overline{Q}_0$ by solving the
loss-based minimization problem for $\varepsilon_n$. By nature of its
construction, this final step fluctuates the initial estimator in a direction
(in model space) that results in an updated estimator that better solves the
EIF estimating equation than its counterpart obtained without reference to the
EIF.

The generality of the TMLE framework leaves many details to the investigator.
For any given target estimand, one must determine the unique EIF of RAL
estimators in the non-parametric model $\M$. In some instances, such
calculations can be labor-intensive and require specialized expertise in
semi-parametric theory. Though accessible guidance on the nature of such
calculations has become available in recent years~\citep{hines2022demystifying,
kennedy2022semiparametric}, the challenge of deriving new estimation procedures
is often commensurate to the complexity of the problem at hand. Further,
although TMLE procedures for target estimands of common interest are widely
available, including for the statistical functionals that identify the
counterfactual mean under treatment~\citep{Gruber2010} (and related contrasts
such as the average treatment effect) or the effects of time-varying treatment
regimens~\citep{robins1986, Stitelman2011, vanderLaan2012, Diaz2021}, it
remains challenging for investigators with limited expertise in semi-parametric
theory to derive TMLE algorithms for novel estimands or for well-studied
estimands in settings subject to complex forms of censoring or sampling.

Concurrent developments in causal machine learning have formulated alternative
representations for the EIFs of common classes of target estimands. Consider
statistical estimands of the form \ijb{$\Psi(\eta_0; \P_0) =
\E_{\P_0}[h(O; \eta_0)]$, where $\eta_0 \coloneqq \eta_{\P_0}$ are nuisance
functions that depend on $\P_0$ (e.g., $\overline{Q}_0$) and $h$ is a fixed,
linear functional of $\eta$. We stress that the target parameter depends
on $\P_0$ only through the nuisance function $\eta_0$ and the outer
expectation; accordingly, we emphasize this by using the shorthand
$\Psi(\eta_0) \equiv \Psi(\eta_0; \P_0)$}. Among recent
developments,~\cite{Hirshberg2021} and~\cite{Chernozhukov2022} recognized that
when the target estimand admits the form above, the EIF $\phi$ may be expressed
\begin{equation}
    \phi(\P)(O) = h(X; \eta) + \alpha(X)(Y - \eta(X)) - \Psi(\eta) \ ,
\label{eqn:riesz-rep}
\end{equation}
where $X$ denotes all variables in $O$ upon which the nuisance function $\eta$
depends. Here, $\alpha$ is called a \textit{Riesz representer} and is
analogous to the inverse probability weights~\citep{horvitz1952generalization}
and balancing weights~\citep{Zubizarreta2015}, common in causal inference and
missing data problems. Such a representation is motivated by the convenience of
estimating $\alpha$ using specialized machine learning
algorithms~\citep{chernozhukov2022riesznet} or de-biasing well-known regression
algorithms~\citep{BrunsSmith2025}. However, as~\cite{Chernozhukov2022c}
and~\cite{Williams2025} argue based on the work of~\cite{newey1994asymptotic},
this representation simplifies the process of deriving an EIF, especially in
settings where the estimand is complex, such as the iterated conditional
expectation functionals that arise in longitudinal problems subject to
time-varying treatment-confounder feedback~\citep{hernan2024, vdl2018targeted,
robins1986}.

This general class of expressions now forms the basis of the \textit{Riesz
learning} framework~\citep{Williams2025, Chernozhukov2022c}. While many of the
above works focus on direct estimation of the EIF, none, to our knowledge,
construct a TMLE. We argue that the TMLE framework can be adapted to yield
simpler algorithms for estimator construction when the target estimand can be
expressed in terms of an appropriately constructed Riesz representer.

\subsection{Our contributions}

\ijb{We seek to connect TMLE to the paradigm of de-biased machine learning
based on Riesz representers~\citep[e.g.,][]{Chernozhukov2022b,
Chernozhukov2022}. To this end, we pursue two goals. The first is to provide an
intuitive derivation of the EIF in terms of Riesz representers, which we term
the ``Riesz EIF.'' This formula has been previously described for sequences of
nested conditional means \citep{Williams2025, Chernozhukov2022c}; however, our
construction strategy more readily allows (1) proving semiparametric efficiency
in a nonparametric model, in a manner mirroring the existing TMLE literature,
and (2) extending theory to settings such as two-phase sampling and nuisance
parameters other than the conditional mean. Overall, this refinement, we
believe, simplifies the process of deriving EIFs in comparison to classical
strategies for doing so. In addition, our strategy connects many existing
results in causal machine learning, providing an intuitive conceptual
representation that can be easily communicated.}

\ijb{The second goal is to construct novel TMLEs that use this general
recursive Riesz EIF.} Using this EIF, we provide TMLE algorithms that
fluctuate each nuisance parameter by employing its corresponding Riesz
representer as a ``clever covariate.'' We implement our procedures in an
open-source software package, \texttt{\{RieszCML\}}, which unifies estimation
across causal inference settings ranging from longitudinal data to mediation
analysis. \ijb{The provided algorithms and software provide convenient,
largely automated computation for implementing causal machine learning methods.
This serves to connect TMLE to the ongoing research program in de-biased
machine learning with Riesz representers.}

\subsection{Outline}
The remainder of the manuscript is organized as follows.
Section~\ref{sec:linear-functional}
describes the de-biasing properties of the Riesz representer $\alpha$, when it
exists. Next, Section~\ref{sec:eif}
provides a recursively applicable ``Riesz EIF,'' generalizing previous works in
which the parameter mapping $\Psi(\P_0)$ depends on only a single conditional
expectation as a nuisance function, and extends it to two-phase sampling
designs. Section~\ref{sec:tmle} describes TMLE algorithms for these settings,
and Section~\ref{sec:examples} demonstrates how they can
be applied to a few commonly studied estimands. Then,
Section~\ref{sec:simulation} presents
simulation experiments comparing these TMLEs with Riesz representers to a
standard TMLE of the average treatment effect, and
Section~\ref{sec:data-analysis}
illustrates use of the proposed methodology by its application for immune
correlates analysis of the HVTN 505 vaccine efficacy trial (NCT00865566) and
by a simulation study of stochastic-interventional vaccine efficacy calibrated
to the CoVPN COVE trial. We conclude in Section~\ref{sec:discussion}
with a discussion of productive avenues for future
investigation.

%%%%%%%%%%%%%%%%%%%%%%%%%%%%%%%%%%%%%%%%%%%%%%%%%%%%%%%%%%%%%%%%%%%%%%%%%%%%%%%
\section{Mathematical preliminaries and notation}\label{sec:linear-functional}

\subsection{Observed data and statistical model}

Throughout, we consider data on $n$ study units, $O_1, \ldots, O_n$, where the
data on a given study unit is represented by the random variable $O$. We allow
details of the hypothetical data-generating study to vary across examples, and,
accordingly, allow constituent random variables contained within $O$ to vary
across study contexts. Following standard conventions in causal inference, $O$
is composed of a sequence of temporally ordered observations of several random
variables. In a study with treatment occurring at a single time point, we have
$O = (L, A, Y)$, with baseline confounders $L$, treatment $A$, and outcome $Y$.
In a study with time-varying confounders and treatment, we have $O =
(L_0, L_1, A_1, \ldots, L_T, A_T, Y)$, where, for $t = 1, \ldots, T$ time
points, $L_t$ are time-varying confounders and $A_t$ are time-varying
treatments, with $L_0$ additionally denoting time-fixed confounders (e.g.,
sex-at-birth) and $Y \coloneqq L_{T+1}$ denoting a realization of the outcome
variable at the time point following (i.e., $T+1$) the final treatment $A_T$.
We introduce additional data structures, in which $O$ contains alternative
constituent random variables, as needed in what follows.

We assume that $O \sim \P_0$, where $\P_0 \in \M$ is the true and unknown
data-generating law, assumed only to be contained in a non-parametric
statistical model $\M$. Further, we assume that the $n$ copies of $O$ are
obtained by i.i.d.~sampling from $\P_0$, and denote by $\P_n$ the empirical
distribution of the study sample $O_1, \ldots, O_n$.
% At times, when the context is clear, we will for convenience use $\P_n$ also
% to denote the empirical mean---that is, $\P_n f \coloneqq \frac{1}{n}
% \sum_{i=1}^n f(O)$ for some arbitrary function $f(\cdot)$ of the study data
% $O$.
We let the mapping $\Psi(\cdot)$ denote a statistical target parameter---often,
a functional that, under assumptions, identifies a causal estimand. The
(statistical) target parameter is a mapping from the model $\M$ containing the
data-generating law $\P_0$ to an appropriate outcome space (e.g., $\R^d$), that
is $\Psi: \M \to \R^d$. When context is clear, we will also suppress dependence
of $\Psi(\cdot)$ on the law $\P_0$, writing $\psi_0 \coloneqq \Psi(\P_0)$ and
$\psi \coloneqq \Psi(\P)$ for arbitrary $\P \in \M$.

As noted in the preceding section, very often, the parameter $\Psi(\cdot)$  may
be expressed as an integral of a nuisance function $\eta$, which itself depends
on $\P \in \M$; this allows us to write $\Psi(\P) \equiv \Psi(\eta; \P)
\equiv\Psi(\eta) $ and, adopting this convention, we use $\eta$ to denote
nuisance functions that depend on $\P$ and appear in $\Psi(\cdot)$.

\subsection{Riesz representer preliminaries}

Consider a function $f \in \mathcal{H}$, where $\mathcal{H}$ is a Hilbert space
equipped with inner product $\langle \cdot, \cdot\rangle$; for simplicity, we
limit our discussion to Hilbert spaces with inner products mapping to $\R$.
% A \textit{functional} $\psi$ is defined as a function of another function.
Call the functional $\Psi(\cdot)$ \textit{linear} if $\Psi(f) + \Psi(g) =
\Psi(f + g)$ and $\Psi(a \cdot f) = a \cdot \Psi(f)$ for all $f, g \in
\mathcal{H}$ and scalars $a \in \R$, and call $\Psi(\cdot)$ \textit{bounded}
if, for all $f \in \mathcal{H}$, $\lvert\Psi(f)\rvert < M\lVert
f\rVert_{\mathcal{H}}$ for some positive scalar $M \in \R^+$ that is bounded
$M < \infty$. For a bounded linear functional $\Psi(\cdot)$, the Riesz
representation theorem guarantees that there exists a unique function $\alpha$
such that
\begin{equation}\label{eq:riesz-representation}
    \Psi(f) = \langle \alpha, f \rangle \ .
\end{equation}
In Equation~\eqref{eq:riesz-representation}, $\alpha$ is
referred to as a \textit{Riesz representer}.
% It plays a central role in statistics, where the goal is usually to estimate
% $\psi$ from data.

When $\mathcal{H} \equiv L_2(\P)$, the Hilbert space of square-integrable
functions, the inner product with which the space is equipped denotes
integration with respect to the probability measure $d\P$, that is,
\begin{equation*}
    \langle f, g\rangle = \int f g d\P \ .
\end{equation*}
% Probability is then defined as the integral with respect to some probability
% measure $d\P$. Conveniently, integrals also constitute bounded linear
% functionals.
This representation is often useful when studying a functional $\Psi(\eta)$
that integrates a nuisance function $\eta$ over some unknown probability
distribution $d\P^\star$. Consider a setting where integration with respect to
the probability measure $d\P$ is possible (i.e., sampling from $\P$ is
possible) but that $d\P^\star$ is not readily accessible (i.e., we do not
observe any samples from the corresponding $\P^\star$). When the target
parameter depends on $\P^\star$, we may use the Riesz representation theorem to
re-express it in terms of $\P$ as follows:
\begin{equation}\label{thm:change-of-measure}
    \Psi(\eta) = \int \eta d\P^{\star} = \int \eta \frac{d\P^\star}{d\P}d\P =
      \int \alpha
    \eta d\P = \langle \alpha, \eta \rangle \ ,
\end{equation}
where $\alpha = d\P^{\star} / d\P$, the Riesz representer, is the Radon-Nikodym
derivative for a change of measure. Recognizing this form can allow one to
easily determine the Riesz representer of a functional of interest.

\subsection{Riesz representers in semi-parametric efficiency theory}

Often, we take as our goal the statistical estimation problem of developing
estimators of a statistical estimand that is a functional of the true but
unknown data-generating law $\P_0$ only through a nuisance function $\eta_0$,
so that we can express the target parameter $\Psi(\P_0)$ as $\Psi(\eta_0)$. As
example, consider the problem of estimating the statistical functional that
identifies the counterfactual mean of the outcome $Y$ under treatment level
$A = a$ from a collection of study units, sampled i.i.d.~from $\P_0$, with the
individual-level data structure $O = (L, A, Y)$. Application of the
g-formula~\citep{hernan2024} yields the plug-in estimand $\Psi(\P_0) =
\Psi(\eta_0) = \E_0[\E_0[Y \mid A = a, L]]$, where the nuisance function is
$\eta_0(A, L) = \E_0[Y \mid A, L]$. We will return to this example to build
intuition.

Two decades ago,~\cite{van2006targeted} proposed a TMLE that considered
functionals of the entire data distribution $\P_0$. More often, though, we are
interested in functionals, as in the example above, for which the nuisance
function $\eta$ is itself a regression function, such as a conditional
expectation function $\E_0[Y \mid A, L]$. In such instances, we may proceed by
constructing an initial estimator $\eta_n$ of $\eta_0$ via a data-adaptive or
non-parametric regression procedure that minimizes an appropriately selected
loss function~\citep{vdl2004asymptotic, dudoit2005asymptotics}. In such
problems, we can construct consistent and asymptotically normal (CAN) estimators
of the target parameter using the unique Riesz representer, which we may combine
with targeting procedures developed in the TMLE
literature~\citep{van2011targeted}.

\textbf{De-biasing property}. Suppose our goal is to estimate the functional
$\Psi(\eta_0)$ and that we obtain an initial estimator $\eta_n$ via some
data-adaptive or non-parametric regression procedure. Being non-parametric,
this nuisance estimator will achieve a slower-than-parametric rate, such as
$o_\P(n^{-1/4})$. Consequently, our estimator $\psi_n \coloneqq \Psi(\eta_n)$
will typically exhibit an asymptotically non-negligible bias. The bias may be
expressed in terms of the Riesz representer $\alpha$ and the difference between
the nuisance estimator $\eta_n$ and the nuisance function $\eta_0$:
\begin{equation*}
    \Psi(\eta_n) - \Psi(\eta_0) = \Psi(\eta_n - \eta_0) =
      \langle\alpha, \eta_n - \eta_0\rangle \ .
\end{equation*}
Hence, adding the correction term $\langle\alpha, \eta_0 - \eta_n\rangle$
to $\Psi(\eta_n)$ removes this bias exactly. Although the correction depends
on the unknown $\eta_0$, it can often be estimated: when $\eta_0(X) = \E(Y
\mid X)$ is a regression function and $\alpha$ is a function of $X$, we have
$\langle\alpha, \eta_0 - \eta_n\rangle = \E[\alpha(X)(Y - \eta_n(X))]$,
which may be estimated by its empirical mean. If the form of $\alpha$ is known,
the resulting estimator is the well-studied one-step de-biased
estimator~\citep{bickel1993efficient}.

\textbf{Double-robustness property}. When the form of $\alpha$ is unknown, as
may be the case, we can construct an estimator $\alpha_n$ of the underlying
Riesz representer $\alpha_0$ and use $\alpha_n$ to de-bias the initial
estimator $\Psi(\eta_n)$. In this case, the resulting ``second-order'' bias is
\begin{equation*}
    \Psi(\eta_n) + \langle\alpha_n, \eta_0 - \eta_n\rangle  - \Psi(\eta_0) =
      \langle\alpha_n - \alpha_0, \eta_0 - \eta_n\rangle \ ,
\end{equation*}
where the bias will vanish if the product of the errors between $(\alpha_n,
\alpha_0)$ and $(\eta_n, \eta_0)$ is $o_{\P}(n^{-1/2})$. This product-of-errors
structure, which arises often in statistical estimation problems rooted in
causal inference, has been termed a \textit{doubly-robust} remainder. Applying
this to our example of the treatment-specific mean $\Psi(\eta_0) = \E_0[\E_0[Y
\mid A = a, L]]$, the nuisance is $\eta_0(A, L) = \E_0[Y \mid A, L]$ and the
Riesz representer is the inverse probability weight (IPW) $\alpha(A, L) =
\1(A = a)/g(A=a \mid L)$ where $g$ denotes the probability of $A$, conditional
on $L$ (the propensity score of~\cite{rosenbaum1983}). In this case, the
second-order bias takes the common form $\int (\alpha_n - \alpha_0)(\eta_0 -
\eta_n)d\P_0$. This remainder will converge to zero if either the outcome
regression $\eta_n$ or the inverse probability weights $\alpha_n$ are correctly
specified (model double-robustness) or if their rate product of differences
decays faster than $1/\sqrt{n}$ (rate double-robustness).

\textbf{Efficiency property}. The existence of a Riesz representer implies
pathwise differentiability of the target parameter~\citep{newey1994asymptotic,
Hirshberg2021, Chernozhukov2022b}. For the Riesz representer to exist,
the functional must be bounded by the norm of $\eta$; that is,
$\lvert\Psi(\eta)\rvert \leq M \lVert \eta \rVert_{L_2(\P_0)}$ for a finite
positive constant $M$. Consequently, RAL estimators of the target parameter
have a unique EIF in the non-parametric model $\M$, whose variance is the
semi-parametric efficiency bound; under regularity conditions on the nuisance
estimators, estimators attaining this bound can be constructed. Next, we review
derivations of EIFs for constructing such an efficient estimator in a few
examples.

% \textit{A note on estimation}. Even if we estimate the nuisances, we typically
% cannot directly evaluate $\psi$ or an inner product -- we must estimate it
% from
% a sample. Suppose that, assuming $\hat{f}$ were known, it were possible to
% construct an asymptotically linear estimator
% \begin{equation*}
%     \hat{\Psi}(f) = \Pr_n h(\eta)
% \end{equation*}
% Since the debiasing term is a covariance, it can be estimated using the
% empirical covariance $\Pr_n \hat\alpha(f - \hat{f})$, meaning the debiased
% estimator is actually
% \begin{equation*}
%     \Pr_n \Big(h(\hat{f}) + \hat{\alpha}[f - \hat{f}]\Big)
% \end{equation*}
% where the inner term minus the functional is the efficient influence function,
% as we show in the next section.

%Just as the propensity score has the ``balancing property'' in binary
%treatments \citep{rosenbaum1983}, and estimators that condition on it remove
%bias and obtain efficiency \citep{Hejazi2023}, the Riesz representer also has
%the same ``balancing property'', albeit for the more general class of linear
%functionals.

%Most works cite \cite[pg. 1356-1357]{newey1994asymptotic} as the source of the
%efficient influence function using Riesz representers. The second cited page
%provides the recursive algorithm given by \cite{Williams2025}. Chapter 1.5.1
%of \cite{vanderLaan2003} describes how linear models fit into the general
%framework of efficient estimating equations using the notion of nuisance
%tangent spaces, but we find that direct computation may be more illustrative. 

\section{Riesz representers for efficient estimation}\label{sec:eif}

We return to the setting of having observed data from a cohort study comprised
of $n$ study units $O_1, \ldots, O_n$, where the study units are sampled
i.i.d.~from $\P_0$, that is $O \sim \P_0$ and $O = (L, A, Y)$. To construct an
efficient estimator of $\psi_0$, we begin with the class of regular
asymptotically linear (RAL) estimators of the estimand $\psi_0$; RAL estimators
$\psi_n$ take the form
\begin{equation*}
  \psi_n - \psi_0 = \frac{1}{n}\sum_{i=1}^n \phi(\P_0)(O_i) +
  o_{\P}(n^{-1/2}) \ ,
\end{equation*}
where $\phi$ is a mean-zero function of finite variance called the
\textit{influence function} of the estimator $\psi_n$; in the non-parametric
model $\M$, there is a unique influence function for the class of RAL
estimators with variance matching the semi-parametric efficiency bound, called
the \textit{efficient influence function}~\citep{van2011targeted}. The notation
$\phi(\P_0)(O)$ is used to indicate that $\phi$ depends on the true data
distribution $\P_0$ and the observed data $O$.

Let us now consider the class of bounded linear functionals, so that the target
parameter can be expressed $\Psi(\eta_0) = \E(h(O; \eta_0))$, where $h$
denotes a transformation of the nuisance function $\eta_0$ consistent with the
parameter definition. Returning to our example, the estimand $\psi_0 =
\E_{\P_0}[\E_{\P_0}[Y \mid A = a, L]]$ is a bounded linear functional $\psi_0
\equiv \Psi(\eta_0) = \E(h(O; \eta_0))$, where $\eta_0 = \E_{\P_0}[Y \mid A,
L]$ and, here, $h(\cdot)$ is the point-evaluation transformation yielding $h(O;
\eta_0) = \E_{\P_0}[Y \mid A = a, L]$. To construct a RAL estimator of such a
functional, we can use the representation above as a guide: an efficient RAL
estimator will (approximately) solve the empirical EIF estimating equation,
that is, $n^{-1}\sum_{i=1}^n \phi(\hat{\P}_n^\star)(O_i; \psi_n) =
o_\P(n^{-1/2})$ for an estimate $\hat{\P}_n^\star$ of the relevant parts of
$\P_0$, and have asymptotic variance equal to the variance of the EIF,
$\E_{\P_0}([\phi(\P_0)(O)]^2)$. Several causal machine learning
frameworks have been proposed for the construction of asymptotically efficient
RAL estimators, including one-step de-biased
estimation~\citep{bickel1993efficient} and targeted minimum loss-based
estimation~\citep{van2011targeted}; all allow for an initial estimator $\eta_n$
of $\eta_0$ to be obtained via data-adaptive, non-parametric regression and
machine learning procedures. All such frameworks require knowledge of the EIF
for the class of RAL estimators of the target parameter. Riesz representation
can be used to simplify the derivation of the EIF. We describe this next.

\subsection{The single-step Riesz EIF}

Let us first consider the case where the nuisance $\eta$ may be estimated
directly from the observed data on $n$ study units, $O_1, \ldots, O_n$.

\begin{theorem}[Riesz EIF]\label{thm:eif-basic}
    Let $\Psi(\eta) = \E[h(O; \eta)]$ be a bounded linear functional of $\eta$
    with Riesz representer $\alpha$. Suppose $\eta = \eta_{\P}$, a function of
    the data law $\P \in \M$. The efficient influence function of
    $\Psi(\eta_0)$ takes the form
    \begin{equation*}
      \phi(\P)(O) = h(O; \eta_{\P})  - \Psi(\eta_{\P}) +
        \int \alpha(O) \phi_{\eta}(\P)(O) d\P \ ,
    \end{equation*}
    where $\phi_\eta(\P)(O)$ is the pathwise derivative of the
    nuisance function $\eta$.
\end{theorem}

Here, $h$ should be thought of as the linear functional of the nuisance parameter that defines the estimand of interest. For example, for the causal effect of a static intervention (e.g., setting $A = 1$ or $A = 0$), $h$ is the evaluation functional that obtains the value of $\eta$ at the relevant treatment level. 

\begin{proof}
We will use a pointwise contamination strategy~\citep{hines2022demystifying,
Ichimura2022}, coupled with the Riesz representation identity. Let $\delta_O$
denote a Dirac delta function that places all mass at the observed $O$. With
the pointwise contamination approach, we construct a path $\P_\varepsilon =
\varepsilon\delta_O + (1 - \varepsilon)\P_0$ through the data-generating law
$\P_0 = \P_{\varepsilon=0}$ and compute its pathwise derivative:

\begin{align*}
    \frac{\partial}{\partial \varepsilon}
      \Psi(\eta_{\P_\varepsilon};\P_\varepsilon)
      \Big|_{\varepsilon = 0} &= \int\frac{\partial}{\partial \varepsilon} h(O;
      \eta_{\P_\varepsilon})d\P_\varepsilon\Big|_{\varepsilon = 0}\\
    &= \int\Big(\frac{\partial}{\partial \varepsilon} h(O;
      \eta_{\P_\varepsilon})
      \Big|_{\varepsilon = 0}\Big)d\P_0 + \int h(O;
        \eta_{\P_0})\frac{\partial}{\partial \varepsilon}
      \P_\varepsilon\Big|_{\varepsilon = 0}\tag{Product rule}\\
    &= \int h\Big(O; \frac{\partial}{\partial \varepsilon}\eta_{\P_\varepsilon}
      \Big|_{\varepsilon = 0}\Big)d\P_0 + \int h(O; \eta_{\P_0})(\delta_O -
        d\P_0)\tag{Linearity of $h$ and $\frac{\partial}{\partial \varepsilon}$;
        evaluate $\frac{\partial}{\partial \varepsilon} \P_\varepsilon$ at
        $\varepsilon = 0$}\\
    &= \int\alpha(X)\underbrace{\Big(\frac{\partial}{\partial \varepsilon}
      \eta_{\P_\varepsilon}(O)\Big|_{\varepsilon =
      0}\Big)}_{\substack{\text{pathwise derivative}\\\text{of nuisance}}}d\P_0
        +
      \underbrace{h(O; \eta_{\P_0}) -
        \Psi(\eta_{\P_0})}_{\substack{\text{pathwise derivative}\\\text{of
        mean}}} \ . \tag{Apply Riesz representation and simplify}
\end{align*}

\noindent The proof concludes by noting that $\frac{\partial}{\partial
\varepsilon} \eta_{\P_\varepsilon}(O)\Big|_{\varepsilon = 0}$ is, by
definition, the pathwise derivative $\phi_\eta(\P_0)(O)$ of the nuisance
$\eta_\P$.
\end{proof}

To summarize, we compute the EIF by calculating the pathwise derivative of our
desired parameter using the product rule, and then applying the Riesz
representation theorem. Our strategy was to delay use of the identity until we
had obtained a form involving integration with respect to the true
data-generating law $\P_0$. Other than Riesz representation, the technique used
to derive this result follows existing strategies from semi-parametric theory,
chiefly the computation of a pathwise derivative
\citep{kennedy2016semiparametric, kennedy2022semiparametric, van2011targeted}.
Its form mirrors familiar results on efficient influence functions such as
those of the identifying statistical functionals of the counterfactual mean and
the average treatment effect. \ijb{A technical note: While $\phi_\eta$ is
occasionally referred to as the ``EIF'' of the nuisance $\eta$, we refer to it
only as a pathwise (G\^ateaux) derivative because it may lie outside the space
$L_2(\P)$, e.g., due to the appearance of a Dirac delta, in
which case $\eta$ would not be pathwise differentiable. Regardless, the
parameter $\Psi(\eta)$ will remain pathwise differentiable because $\phi_\eta$
is integrated with respect to $\P$ and any Dirac delta functions are thus
integrated out.}

The result generalizes Riesz-based EIF representations presented by previous
authors, including \cite{newey1994asymptotic}, \cite{Hirshberg2021}, and
\cite{Chernozhukov2022} in the sense that any nuisance function may be
considered, not just conditional expectation functions. \ijb{A similar idea is
used by \cite{vdl2025} to derive the EIF of linear functionals where $\eta$ is
defined as the minimizer of an appropriate loss function.}

When $\eta_{\P}$ is a regression function, its pathwise derivative is often
known or may otherwise be easily obtained. For example, consider the
conditional expectation function $\overline{Q}(A, L) = \E_{\P}(Y \mid A, L)$
with pathwise derivative:
\begin{equation*}
  \frac{\delta_{L, A}}{d\P(A, L)}[Y - \overline{Q}(A, L)] \ ,
\end{equation*}
where $\delta$ is a Dirac delta function that places all mass at the
observed $(L, A)$ and zero elsewhere; see \cite{hines2022demystifying} for a
discussion. As integration over a Dirac delta function corresponds to point
evaluation, we have the following corollary:
\begin{corollary}[Riesz EIF with conditional mean]\label{cor:eif-condmean}
    Let $\Psi(\overline{Q}) = \E_{\P}[h(O; \overline{Q})]$ be a bounded linear
    functional of the conditional mean function $\overline{Q}(A, L) =
    \E_{\P}(Y \mid A, L)$, with Riesz representer $\alpha$. Then, the EIF of
    $\Psi(\overline{Q})$ is
    \begin{equation*}
      \phi(\P)(O) = h(O; \overline{Q}) - \psi + \alpha(A, L)(Y -
      \overline{Q}(A, L)) \ .
    \end{equation*}
\end{corollary}
This Riesz EIF recovers the form discussed at length by previous authors
\citep[e.g.,][]{newey1994asymptotic, Hirshberg2021, Chernozhukov2022,
Williams2025}. It mirrors familiar presentations of the EIF for parameters like
the identifying functional of the treatment-specific mean, whose EIF is
\begin{equation*}
    \phi(\P)(O) = \overline{Q}(a, L) - \psi+ \frac{\1(A = a)}{g(a \mid L)}
      \Big(Y - \overline{Q}(A, L)\Big) \ ,
\end{equation*}
where $g(a \mid L) \coloneqq \Pr_{\P}(A = a \mid L)$, the propensity score.

Theorem~\ref{thm:eif-basic} can also be applied
\textit{recursively}. To assist with this, we first introduce a preliminary
lemma.
\begin{lemma}[Conditional Riesz EIF]\label{lem:eif-basic-cond}
    Let $V \subset O$ denote a subset of variables. Consider the conditional
    estimand $\Psi_v(\eta) = \E[h(O; \eta) \mid V = v]$. Its EIF is
    \begin{equation*}
        \E[\alpha(O)\phi_\eta(\P)(O) \mid V = v] + \frac{\delta_v}{d\P_V(v)}
        \Big(h(O; \eta_\P) - \E[h(O; \eta_\P) \mid V = v]\Big) \ ,
    \end{equation*}
    where $\phi_\eta(\P)(O)$ denotes the pathwise derivative of the nuisance
    $\eta$, $\delta_v$ a Dirac delta function placing all mass at $V = v$, and
    $d\P_V$ the density of $V$.
\end{lemma}

\begin{proof}
    Using the point mass contamination strategy, we obtain the result by
    applying the product rule to the EIF for a conditional mean functional,
    following Example 6 of \cite{hines2022demystifying} to note
    \begin{align*}
      \frac{\partial}{\partial \varepsilon}\Psi_v(\eta_{\P_\varepsilon};
      \P_\varepsilon) =~&\E\Big[h\Big(O; \frac{\partial}{\partial \varepsilon}
        \eta_{\P_\varepsilon} \Big|_{\varepsilon = 0}\Big) \mid V = v\Big] \\&+
        \frac{\delta_v}{d\P_{0,V}(v)}\Big(h(O; \eta_{\P_0}) -
        \E[h(O; \eta_{\P_0}) \mid V = v]\Big) \ ,
    \end{align*}
    and then, applying the Riesz representation theorem to the first term, we
    note that the derivative of $\eta_{\P_\varepsilon}$ at
    $\varepsilon = 0$ is the pathwise derivative $\phi_\eta(\P_0)(O)$.
\end{proof}

Unless $V$ is discrete, this parameter is not pathwise differentiable due to
the presence of the Dirac delta (hence, we only employ Riesz representation for
the first component under expectation). However, it will be a useful
intermediate result for applying Theorem~\ref{thm:eif-basic}
recursively. We apply this lemma in the next two sections to derive Riesz-based
EIFs for more complex parameters of interest.

\subsection{Sequential representation of a Riesz EIF}\label{sec:sequential}

Suppose now, rather than data from a study with treatment administered at a
single time point, one instead observes a sequence of time-ordered variables,
where the data unit may be represented $O = (L_0, L_1, A_1, L_2, A_2, \ldots,
L_T, A_T, Y)$. The prototypical example is a longitudinal study with
time-varying treatment and time-varying treatment-confounder
feedback~\citep{robins1986, hernan2024}, but this data unit and structure may
also encompass settings like causal mediation analysis, for example. As before,
we assume that we have access to an i.i.d.~sample of $n$ study units, $O_1,
\ldots, O_n$, where each $O \sim \P_0 \in \M$.

In such a setting, we can consider an estimand represented by a sequence of
regression functions. For instance, updating our running example, a common goal
in causal inference is to evaluate the counterfactual mean under a fixed
treatment sequence, say, $(A_1 = a, A_2 = a, \ldots, A_T = a)$, indicating that
treatment is set to $a \in \mathcal{A}$ at all time points $t = 1, \ldots, T$.
The effect of such a treatment regime may be identified via the
g-formula~\citep{robins1986, hernan2024} and leads to an estimand of the form
\begin{equation*}
\Psi(\P_0) =
  \E[\E[\cdots \E[\E[Y \mid \bar{A}_T
    = \bar{a}_T, \bar{L}_T] \mid \bar{A}_{T-1} =
      \bar{a}_{T-1}, \bar{L}_{T-1}] \cdots \mid A_1 = a, L_1, L_0]] \ ,
\end{equation*}
where we use $\bar{Z}_t$ to denote the history of random variable $Z$ up until
time $t$, that is, $\bar{Z}_t = (Z_0, Z_1, \ldots, Z_{t})$, where we set
$A_0 = \emptyset$ since no treatment is administered at baseline, so that the
history $\bar{L}_t$ includes the baseline covariates $L_0$; each expectation
is taken over $\P_0$. Because a conditional expectation is itself a linear
functional, the corresponding EIF of such an estimand may be expressed in terms
of sequential expressions involving Riesz representers, as elaborated upon
next.

\begin{theorem}[Sequential Riesz EIF]\label{thm:eif-recursive}
    Consider the bounded linear functional 
    $$\Psi(\P) = \E_{\P}[h_1(\bar{A}_{1}, \bar{L}_{1};
    \overline{Q}_1)]\,$$
    where the nuisance functions $\overline{Q}_t$ are defined sequentially
    such that, for $t = 1, \ldots, T$, we have
    \begin{equation*}
      \overline{Q}_{t}(\bar{A}_{t}, \bar{L}_{t}) = \E_\P[h_{t+1}(\bar{A}_{t+1},
      \bar{L}_{t+1}; \overline{Q}_{t+1}) \mid \bar{A}_t, \bar{L}_t] \ ,
    \end{equation*}
    with the convention $h_{T+1}(\bar{A}_{T+1}, \bar{L}_{T+1};
    \overline{Q}_{T+1}) \coloneqq Y$, so that $\overline{Q}_T(\bar{A}_T,
    \bar{L}_T) = \E_\P(Y \mid \bar{A}_T, \bar{L}_T)$. Suppose that
    $f \mapsto \E_\P[h_1(\bar{A}_1, \bar{L}_1; f)]$ and, for
    $t = 1, \ldots, T-1$, $f \mapsto \E_\P[h_{t+1}(\bar{A}_{t+1},
    \bar{L}_{t+1}; f) \mid \bar{A}_t, \bar{L}_t]$ are bounded linear
    functionals. Let $\alpha_1$ denote the Riesz representer for
    $\overline{Q}_1$ in $\Psi$, that is, $\E_\P[h_1(\bar{A}_1, \bar{L}_1; f)] =
    \E_\P[\alpha_1(\bar{A}_1, \bar{L}_1) f(\bar{A}_1, \bar{L}_1)]$, and,
    for $t = 1, \ldots, T-1$, let $\alpha_{t+1}$ denote the Riesz representer
    for $\overline{Q}_{t+1}$ in the functional defining $\overline{Q}_t$, that
    is, $\E_\P[h_{t+1}(\bar{A}_{t+1}, \bar{L}_{t+1}; f) \mid \bar{A}_t,
    \bar{L}_t] = \E_\P[\alpha_{t+1}(\bar{A}_{t+1}, \bar{L}_{t+1})
    f(\bar{A}_{t+1}, \bar{L}_{t+1}) \mid \bar{A}_t, \bar{L}_t]$. The EIF of
    $\Psi(\P)$ is
    \begin{align*}
      \phi(\P)(O) =&~\sum_{t=1}^{T}\prod_{k=1}^t \alpha_k(\bar{A}_{k},
        \bar{L}_{k})[h_{t+1}(\bar{A}_{t+1}, \bar{L}_{t+1};
        \overline{Q}_{t+1}) - \overline{Q}_t(\bar{A}_{t}, \bar{L}_{t})] \\
        &+ h_1(\bar{A}_{1}, \bar{L}_{1}; \overline{Q}_1) - \psi \ .
    \end{align*}
\end{theorem}

\begin{proof}
We apply Lemma~\ref{lem:eif-basic-cond}
recursively over a changing sequence of
conditional distributions. Note that, by the sequential nature of the data
structure, the probability measure of the data can be decomposed as
\begin{equation*}
  d\P = d\P(Y \mid \bar{A}_T, \bar{L}_T)
    \prod_{t=1}^T d\P(A_t, L_t \mid \bar{A}_{t-1}, \bar{L}_{t-1}) \,
    d\P(L_0) \ .
\end{equation*}
To ease notational burden, we denote $d\P^t = d\P(A_t, L_t \mid
\bar{A}_{t-1}, \bar{L}_{t-1})$ for $t = 2, \ldots, T$ and $d\P^1 = d\P(A_1, L_1,
L_0)$ as shorthand, so that $d\P(\bar{A}_t, \bar{L}_t) = \prod_{k=1}^t d\P^k$.
We will also denote the EIF of $\overline{Q}_{t}(\bar{A}_{t}, \bar{L}_{t})$ as
$\phi_{\overline{Q}_t}(\P_0)(O)$. The proof will consist of three components:
(1) the EIF of the first step $t = 1$, (2) the EIF of the last step $t = T$,
and (3) the EIF of intermediate steps.

\noindent (1) In the first case, at $t = 1$, Theorem~\ref{thm:eif-basic}
implies
\begin{align*}
  \phi(\P)(O) =~& h_1(A_{1}, L_{1}, L_0; \overline{Q}_1) - \psi_0 \\ 
    &+ \int\alpha_1(A_1, L_1, L_0) \phi_{\overline{Q}_1}(A_1, L_1,
    L_0)d\P_0(A_1, L_1, L_0) \ .
\end{align*}

\noindent (2) In the second case, at $t = T$, the EIF of the final regression
$\overline{Q}_T(\bar{A}_T, \bar{L}_T) \coloneqq \E(Y \mid \bar{A}_T,
\bar{L}_T)$ follows from that of a conditional mean:
\begin{equation*}
    \phi_{\overline{Q}_T}(\P)(O) = \frac{\delta_{\bar{A}_T, \bar{L}_T}}
    {d\P(\bar{A}_T, \bar{L}_T)}
    \left[Y - \overline{Q}_T(\bar{A}_T, \bar{L}_T)\right] \ .
\end{equation*}

\noindent (3) In each intermediate step, applying
Lemma~\ref{lem:eif-basic-cond}, we have that
the EIF of an arbitrary element $\overline{Q}_t$ of the sequence is
\begin{align*}
    & \phi_{\overline{Q}_t}(\P)(O) = \int \alpha_{t+1}(\bar{A}_{t+1},
      \bar{L}_{t+1})
    \phi_{\overline{Q}_{t+1}}(\P)(\bar{A}_{t+1}, \bar{L}_{t+1}) d\P^{t+1} \\
    & + \frac{\delta_{\bar{A}_t, \bar{L}_t}}
    {d\P(\bar{A}_t, \bar{L}_t)}\Big(h_{t+1}(\bar{A}_{t+1}, \bar{L}_{t+1};
    \overline{Q}_{t+1}) - \E[h_{t+1}(\bar{A}_{t+1}, \bar{L}_{t+1};
    \overline{Q}_{t+1}) \mid \bar{A}_{t}, \bar{L}_{t}]\Big) \ .
\end{align*}
Finally, plug into EIF (2) the intermediate EIF (3) recursively until $t
= 1$, and then plug that into EIF (1) to obtain the EIF for the final step. The
weights $\delta_{\bar{A}_t, \bar{L}_t}/d\P(\bar{A}_t, \bar{L}_t)$ eliminate each
integral operator by turning the integration into a point-evaluation at
$(\bar{A}_{t}, \bar{L}_{t})$ at each step and canceling out the accumulating
measures $d\P(\bar{A}_t, \bar{L}_t) = \prod_{k=1}^t d\P^k$, leaving the product
$\prod_{k=1}^{t}\alpha_k$ as the weight on the $t$\textsuperscript{th}
residual.
\end{proof}

\begin{remark}\label{rem:product}
    \ijb{The form of this EIF differs slightly from the EIF for recursively
    Riesz-representable estimands presented by~\cite{Chernozhukov2022c}
    and~\cite{Williams2025} in that it features a product of Riesz representers
    $\prod_{k=1}^t \alpha_k$ in place of a single Riesz representer for each
    sequential regression. In the above proof, instead of using
    Theorem~\ref{thm:eif-basic} and
    Lemma~\ref{lem:eif-basic-cond} to
    immediately convert each $h_t$ to a Riesz representer, we could have
    replaced
    \[\int \alpha_{t+1}(\bar{A}_{t+1}, \bar{L}_{t+1})
    \phi_{\overline{Q}_{t+1}}(\P)(\bar{A}_{t+1}, \bar{L}_{t+1}) d\P^{t+1}\]
    \noindent with the expression
     \[\int h_{t+1}[\bar{A}_{t+1}, \bar{L}_{t+1};
    \phi_{\overline{Q}_{t+1}}(\P)] d\P^{t+1}\]
    Then, in the recursive representation, one would instead perform the Riesz
    representation on a composition of functions like so:
    \begin{align*}
      &\E[(h_1 \circ h_2' \circ \ldots \circ h_{t+1}')(\bar{A}_{t+1},
        \bar{L}_{t+1}; \overline{Q}_{t+1})] \\
      &\quad = \E[\omega_{t+1}(\bar{A}_{t+1}, \bar{L}_{t+1})
        \overline{Q}_{t+1}(\bar{A}_{t+1}, \bar{L}_{t+1})] \ ,
    \end{align*}
    where we define $h_{k}' \coloneqq \overline{Q} \mapsto \E_\P[h_{k}(\bar{A}_{k}, \bar{L}_{k}; \overline{Q}) \mid \bar{A}_{k-1}, \bar{L}_{k-1}]$
    and where $\omega_{t+1}$ exists provided that the left-hand side, as a map
    of $\overline{Q}_{t+1}$, is a bounded linear functional on $L_2(\P)$. As a reminder, $h_1$ is the outermost linear functional of $\overline{Q}$ that defines the estimand. This
    yields the same expression as given for the EIF
    by~\cite{Chernozhukov2022c} and~\cite{Williams2025}
    \begin{align*}
      \phi(\P)(O) =&~\sum_{t=1}^{T} \omega_t(\bar{A}_{t},
        \bar{L}_{t})[h_{t+1}(\bar{A}_{t+1}, \bar{L}_{t+1};
        \overline{Q}_{t+1}) - \overline{Q}_t(\bar{A}_{t}, \bar{L}_{t})]\\
        &+ h_1(\bar{A}_{1}, \bar{L}_{1}; \overline{Q}_1) - \psi \ .
    \end{align*}%
}%
\end{remark}%
\ijb{Indeed, whenever the $\alpha_k$ of Theorem~\ref{thm:eif-recursive}
exist, the $\omega_t$ of Remark~\ref{rem:product} coincide with the product
weights $\prod_{k=1}^t \alpha_k$ given later in
Algorithm~\ref{alg:riesz-tmle}, and it may be more numerically stable to
estimate $\omega_t$ directly. Unlike the expression in
Theorem~\ref{thm:eif-recursive}, the EIF in Remark~\ref{rem:product} does not
require that, for $t = 1, \ldots, T-1$, $f \mapsto
\E_\P[h_{t+1}(\bar{A}_{t+1}, \bar{L}_{t+1}; f) \mid \bar{A}_t, \bar{L}_t]$
are bounded linear functionals. This enables its application to a slightly
wider class of functionals, most notably mediation, which does not satisfy
this condition (see Section~\ref{sec:mediation}). However, because
Algorithm~\ref{alg:seq-tmle} weights by partial products of the
stage-specific $\alpha_t$, it applies only to the form of
Theorem~\ref{thm:eif-recursive}, not to every EIF in the form of
Remark~\ref{rem:product}.}

\subsection{Riesz representation for more complex functionals}

Beyond their use in obtaining efficient influence functions directly, the Riesz
representation strategies of Theorems~\ref{thm:eif-basic}
and~\ref{thm:eif-recursive} can also be used
to derive \textit{components} of an EIF for a functional that may not be
linear. A simple example is the identifying functional for the ratio of
counterfactual means
\begin{equation*}
    \Psi(\P) = \frac{\E[\E(Y \mid A = a_1, L)]}{\E[\E(Y \mid A = a_2, L)]} \ ,
\end{equation*}
whose EIF and asymptotic distribution can be derived in terms of Riesz
representers by first applying Corollary~\ref{cor:eif-condmean}
to the numerator and denominator and then using the delta method. In general,
this strategy can be easily applied to any functional $ \psi = b(\psi_1,
\ldots, \psi_k)$ with components $\psi_1, \ldots, \psi_k$ being bounded linear
functionals and $b$ being differentiable and not depending on the law $\P$.

More involved use-cases arise in the contexts of alternative sampling schemes
or incomplete data. Suppose one would like to obtain inference on a parameter
of the distribution $\P_0^X$ of data unit $X \sim \P_0^X$ but only observes a
coarsened data structure $O = (V, \Delta, \Delta X) \sim \P_0$, meaning $X$ is
only observed according to the coarsening
process~\citep{heitjan1991ignorability, gill1997coarsening} underlying the
sampling indicator $\Delta \in \{0, 1\}$. Here, the sampling probability is
defined $\pi_{\Delta}(V) \coloneqq \P(\Delta = 1 \mid V)$ that depends on the
always-observed covariates $V \subset X$ but not on the remainder of $X$ (i.e.,
the data are coarsened at random); we further assume $\pi_{\Delta}(V) > 0$
almost surely. Thus, the observed data law $\P_0$ is implied by the
complete-data law $\P_0^X$ and the coarsening mechanism $\pi_{\Delta}$. In this
case, one can combine Theorem~\ref{thm:eif-basic} and
Lemma~\ref{lem:eif-basic-cond} to derive an EIF analogous to those presented
by~\cite{robins1994estimation} and~\cite{Rose2011}.

\begin{theorem}\label{thm:eif-twophase}
    Let $\Psi$ be a pathwise differentiable parameter whose EIF under
    complete data $X \sim \P^X$ is given as $\phi(\P^X)(X)$. Then, under the
    coarsening conditions above, its EIF under the observed data $O \sim \P$
    takes the form
    \begin{align*}
        \phi(\P)(O) =~& \E[\phi^{\text{uc}}(\P^X)(X) \mid \Delta= 1, V ] -
          \psi \\&+ \alpha(\Delta, V)(\phi^{\text{uc}}(\P^X)(X) -
          \E[\phi^{\text{uc}}(\P^X)(X) \mid \Delta, V]) \ ,
    \end{align*}
    where $\phi^{\text{uc}}(\P^X)(X)$ denotes the uncentered complete-data
    EIF $\phi(\P^X)(X) + \psi$ and $\alpha(\Delta, V) = \Delta /
    \pi_{\Delta}(V)$ is the Riesz representer of the functional $\eta \mapsto
    \E[\eta(1, V)]$.
\end{theorem}

\begin{proof}
    This result follows by noting that
    \begin{equation*}
      \psi = \E[\phi^{\text{uc}}(\P^X)(X)] =
        \E(\E[\phi^{\text{uc}}(\P^X)(X) \mid \Delta = 1, V]) \ ,
    \end{equation*}
    where the second equality uses that the data are coarsened at random, and
    applying Theorem~\ref{thm:eif-basic} with
    $\eta(\Delta, V) = \E[\phi^{\text{uc}}(\P^X)(X) \mid \Delta, V]$ as the
    nuisance function to obtain
    \begin{align}\label{eq:twophase-proof}
      \phi(\P)(O) =~& \E[\phi^{\text{uc}}(\P^X)(X) \mid \Delta= 1, V ] -
        \psi
        \nonumber \\&+ \int\alpha(\Delta, V) \phi_\eta(\P)(\Delta, V)
        d\P(\Delta, V) \ .
    \end{align}
    \noindent
    By Lemma~\ref{lem:eif-basic-cond}, applied with $h$ the identity (so that
    its Riesz representer is one), the nuisance function's pathwise
    derivative $\phi_{\eta}$ is
    \begin{align*}
      \phi_{\eta}(\P)(\Delta, V) =~&
         \E[\dot{\phi}^{\text{uc}}(X) \mid \Delta, V]
         \\&+ \frac{\delta_{\Delta, V}}{d\P(\Delta,
           V)}\Big(\phi^{\text{uc}}(\P^X)(X)
         - \E[\phi^{\text{uc}}(\P^X)(X) \mid \Delta, V]\Big) \ ,
    \end{align*}
    where $\dot{\phi}^{\text{uc}}(X) \coloneqq \frac{\partial}{\partial
    \varepsilon} \phi^{uc}(\P^X_\varepsilon)(X)\big|_{\varepsilon = 0}$ denotes
    the derivative of the complete-data EIF, as a function of the
    data-generating law, along the path $\P^X_\varepsilon$ of complete-data
    laws induced by the point-mass contamination of $\P$. This derivative has
    mean zero. To see this, first consider the point-mass contamination
    $\P^X_\varepsilon = \varepsilon\delta_x + (1 - \varepsilon)\P^X$ of the
    complete-data law at a point $x$. Since $\E_{\P^X}[\phi^{\text{uc}}
    (\P^X)(X)] = \Psi(\P^X)$ for every complete-data law, differentiating both
    sides at $\varepsilon = 0$ with the product rule, as in the proof of
    Theorem~\ref{thm:eif-basic}, gives
    \begin{equation*}
      \E[\dot{\phi}^{\text{uc}}(X)] + \phi^{\text{uc}}(\P^X)(x) - \psi =
      \phi(\P^X)(x) \ ,
    \end{equation*}
    so that $\E[\dot{\phi}^{\text{uc}}(X)] = 0$, since
    $\phi^{\text{uc}}(\P^X)(x) - \psi = \phi(\P^X)(x)$. The path induced by
    contaminating the observed-data law perturbs $\P^X$ in the direction of at
    signed combination of such point masses, so, by linearity of the derivative
    in the direction of perturbation, $\E[\dot{\phi}^{\text{uc}}(X)] = 0$ along
    that path as well. Plugging in the expression for $\phi_\eta(\Delta, V)$
    into the second term of Equation~\eqref{eq:twophase-proof}, we
    observe
    \begin{align*}
      &\int\alpha(\Delta, V) \E[\dot{\phi}^{\text{uc}}(X) \mid \Delta, V]
        d\P(\Delta, V) = \E[\dot{\phi}^{\text{uc}}(X)] = 0
          \tag{$\E[\alpha(\Delta, V)
        \mid X] = 1$ under coarsening at random}\\
      &\int\alpha(\Delta, V) \frac{\delta_{\Delta, V}}{d\P(\Delta, V)}
        \Big(\phi^{\text{uc}}(\P^X)(X) -
        \E[\phi^{\text{uc}}(\P^X)(X) \mid \Delta, V]\Big)d\P(\Delta, V) \\
      &= \alpha(\Delta, V)\Big(\phi^{\text{uc}}(\P^X)(X) -
        \E[\phi^{\text{uc}}(\P^X)(X) \mid \Delta, V]\Big) \tag{Cancel
          probability measures and integrate over Dirac delta} \ .
    \end{align*}
The above steps follow by algebraic cancellation. Plugging these results into
Equation~\eqref{eq:twophase-proof} concludes the proof.
\end{proof}

\section{Riesz TMLE}\label{sec:tmle}

When $\Psi(\cdot)$ is a linear functional of nuisance functions, a TMLE can be
constructed to exploit the Riesz representer. The key idea relied upon by the
TMLE framework is that \textit{any} estimator solving the EIF estimating
equation will, under the conditions of Remark~\ref{rem:asymptotics}, be
asymptotically efficient. TMLE uses this to construct asymptotically efficient
plug-in estimators.

The ``classical'' approach to efficient estimation is to solve the estimating
equation $\E_{\P_n}\phi(\hat{\P}_n)(O; \psi) = 0$ directly for $\psi$.
However, one can take another tack. Suppose we have an initial estimator
$\psi_n$ (such as the plug-in estimator); we may propose to use a loss function
$\mathcal{L}(O; \varepsilon)$ that depends on one or more parameters
$\varepsilon$ such that for some vector $v$ of known constants,
\begin{equation}
  v^\top\frac{\partial}{\partial \varepsilon} \E_{\P_n}[
  \mathcal{L}(O; \varepsilon)]\Big|_{\varepsilon = 0} =
  \E_{\P_n}[\phi(\hat{\P}_n)(O)] \ .
\end{equation}
In this setup, the initial estimate is updated by minimizing the empirical risk
over $\varepsilon$, iterating if necessary until $\varepsilon_n \approx
0$, at which point the updated estimate $\hat{\P}_n^\star$ solves the EIF
estimating equation $\E_{\P_n}\phi(\hat{\P}_n^\star)(O) = 0$; for the
fluctuations used below, a single update suffices~\citep{van2006targeted}.

We follow similar logic to~\cite{van2006targeted} and~\cite{Gruber2010} to
construct a TMLE for functionals with Riesz-based EIFs by ``fluctuating'' the
nuisance estimate, using the Riesz representer as a ``clever covariate''. As a
warm-up, suppose $O = (L, A, Y)$ and our EIF takes the form
\begin{equation*}
    \phi_\text{EIF}(O) = \underbrace{h(A, L; \overline{Q})  -
    \psi}_{\phi_1\text{ Sample mean EIF}} + \underbrace{\alpha(A, L)[Y -
    \overline{Q}(A, L)]}_{\phi_2 \text{ Riesz-based sub-EIF}} \ .
\end{equation*}
We can construct a loss $\mathcal{L} = \mathcal{L}_1 + \mathcal{L}_2$ whose
pathwise derivatives equal $\phi_1$ and $\phi_2$, up to sign. For the first
part, let $\P_{\varepsilon_1} = (1 + \varepsilon_1 \phi_1)\P$, and the loss
$\mathcal{L}_1(\P) = -\log(\P)$, so that $\mathcal{L}_1(\P_{\varepsilon_1}) =
-\log\Big((1 + \varepsilon_1\phi_1)\P\Big)$. Taking the derivative yields
$\frac{\partial}{\partial
\varepsilon_1}\mathcal{L}_1(\P_{\varepsilon_1})\Big|_{\varepsilon_1 = 0} =
-(\phi_1 \P) / [(1 + (0)\phi_1)\P] = -\phi_1$, so optimizing this loss solves
the first part of the estimating equation. Moreover, when the marginal
distribution of $(A, L)$ is estimated by its empirical distribution, the plug-in
estimate of $\psi$ satisfies $n^{-1}\sum_{i=1}^n \phi_1(O_i) = 0$ exactly, so
that $\varepsilon_1 = 0$ minimizes the empirical risk; this part may thus be
excluded from the loss, because its solution is known.

For $\mathcal{L}_2$, we can define $\overline{Q}_{\varepsilon_2} =
\overline{Q} - \varepsilon_2\alpha$, and a loss that, at $\varepsilon_2 = 0$,
equals $\phi_2$. For example, the squared error loss satisfies
\begin{align*}
  \frac{\partial}{\partial \varepsilon_2}\mathcal{L}_2(O;
    \varepsilon_2)\Big|_{\varepsilon_2 = 0} &=
    \frac{\partial}{\partial \varepsilon_2}\frac{1}{2}[Y - \overline{Q}(A, L) +
    \varepsilon_2\alpha(A, L)]^2\Big|_{\varepsilon_2 = 0} \\
    &= \alpha(A, L)[Y - \overline{Q}(A, L)] \ ,
\end{align*}
which is equivalent to $\phi_2$. Since we used the squared error loss,
$\varepsilon_2$ can be estimated by simply fitting linear regression with
offset $\overline{Q}(A, L)$ and outcome $Y$. If $\mathcal{L}_2$ were some other
loss function, we could fit an appropriate regression using that loss. For
example,~\cite{Gruber2010} \ijb{propose fitting a logistic regression, which
optimizes the negative log-likelihood loss, to handle cases where $Y$ is
bounded. In practice this is done by scaling $Y$ to fall in the closed unit
interval, $[0, 1]$, before fitting the regression, and then re-scaling back to
the original scale either after estimating $\overline{Q}(A, L)$ or after
calculating the TMLE}. In general, we will consider regressions of the form
\ijb{$\link(\bar{Q}_{\varepsilon}(A, L)) = \link(\overline{Q}(A, L)) +
\varepsilon \alpha(A, L)$, where $\link(\overline{Q}(A, L))$ on the right-hand
side denotes an offset.}

Now, suppose our EIF was instead defined as in
Theorem~\ref{thm:eif-recursive}. Our goal
should be to construct a separable loss function $\mathcal{L} = \sum_{j=1}^T
\mathcal{L}_j$ such that minimizing each $\mathcal{L}_j$ over its own
fluctuation parameter $\varepsilon_j$ solves the $j$\textsuperscript{th}
component of the recursive EIF estimating equation.

There are a few different ways to do this. One strategy is to construct a
sequential TMLE analogous to that of~\cite{Stitelman2011}
or~\cite{vanderLaan2012}. The procedure proceeds as described in
Algorithm~\ref{alg:riesz-tmle}. This algorithm
provides a TMLE for estimating the general sequential regression functional
$\Psi(\P) = \E_{\P}[h_1(\bar{A}_{1}, \bar{L}_{1}; \overline{Q}_1)]$, where
$\overline{Q}_t$ is recursively defined $\overline{Q}_{t}(\bar{A}_{t},
\bar{L}_{t}) \coloneqq \E[h_{t+1}(\bar{A}_{t+1}, \bar{L}_{t+1};
\overline{Q}_{t+1}) \mid \bar{A}_t, \bar{L}_t]$, and the final regression is
$\overline{Q}_T(\bar{A}_T, \bar{L}_T) \coloneqq \E(Y \mid \bar{A}_T,
\bar{L}_T)$. To estimate this parameter, one first fits each sequential
regression $\overline{Q}_t$ in reverse time order (recursively), and then
similarly estimates the Riesz representers. Then, each sequential regression is
updated by constructing a univariate fluctuation model, whose outcome depends on
the targeted regression from the following time step and whose offset is the
initial estimate of the regression being updated. Finally, predictions are
produced from each fluctuated sequential regression to construct the final
plug-in estimator $\psi_{n}^{\star}$.

\begin{algorithm}
    \caption{Riesz TMLE}
    \label{alg:riesz-tmle}
    \begin{algorithmic}[1]
        \State Fit sequential regressions $\overline{Q}_T, \ldots,
          \overline{Q}_1$ and Riesz representers $\alpha_1, \ldots, \alpha_T$.
        \State For $t = 1, \ldots, T$, compute the weights
          $\omega_t(\bar{A}_{t}, \bar{L}_{t}) = \prod_{k=1}^t
          \alpha_k(\bar{A}_{k}, \bar{L}_{k})$ \ijb{(or, if the $\alpha_k$ do
          not exist, estimate $\omega_t$ directly; see
          Remark~\ref{rem:product})}.
        \State Set $h_{T+1}(\bar{A}_{T+1}, \bar{L}_{T+1};
          \overline{Q}_{T+1}^\star) = Y$.
        \For{$t = T, \ldots, 1$}
        \AlgStep{Using an appropriate link function, fit a one-dimensional GLM
          with outcome $h_{t+1}(\bar{A}_{t+1},\bar{L}_{t+1};
          \overline{Q}_{t+1}^\star)$, offset
          $\text{link}[\overline{Q}_{t}(\bar{A}_{t}, \bar{L}_{t})]$,
          covariate $\omega_t(\bar{A}_t, \bar{L}_t)$, and no intercept.}
        \AlgStep{Update $\text{link}[\overline{Q}_{t}^\star(\bar{A}_t,
          \bar{L}_t)] = \text{link}[\overline{Q}_{t}(\bar{A}_t, \bar{L}_t)] +
          \varepsilon_{n,t}\omega_t(\bar{A}_t, \bar{L}_t)$, where
          $\varepsilon_{n,t}$ (with $j = t$) denotes the estimated coefficient
          from the fitted GLM.}
        \EndFor
        \State Compute the final TML estimate:
          \begin{equation*}
            \psi_{n}^{\star} = \frac{1}{n}\sum_{i=1}^n h_1(\bar{A}_{1},
              \bar{L}_{1}; \overline{Q}_1^\star) \ .
          \end{equation*}
    \end{algorithmic}
\end{algorithm}

 In this TMLE procedure, at each step, the score equation of $\mathcal{L}_t$ is
 solved, meaning it solves the entire recursive estimating equation. It is
 doubly-robust in the sense that consistency follows from consistent estimation
 of \textit{either} all the nuisance estimates of $\bar{Q}_1, \ldots, \bar{Q}_T$
 \textit{or} all the Riesz representers $\alpha_1, \ldots, \alpha_T$, similar to
 the procedure of~\cite{vanderLaan2012} or that of~\cite{Diaz2021}.
 Alternatively, one could also use a common fluctuation coefficient
 $\varepsilon$, as proposed in~\cite{vanderLaan2012}, and iteratively update it,
 but this may require more computational resources. Since $\prod_t\alpha_t$ is
 itself a Riesz representer, one can also estimate $\omega_t$ directly to
 improve stability.

An alternative strategy can achieve the slightly stronger condition of
\textit{sequential} double-robustness, which only requires that either
$\overline{Q}_t$ \textit{or} $\alpha_t$ be estimated consistently at a given
time step. This condition is explored by~\cite{Luedtke2017} and discussed
by~\cite{Diaz2021}. It can be achieved by fitting fluctuation
\textit{functions} $\varepsilon_{n,t}^s(\bar{A}_s, \bar{L}_s)$ of the
histories $s \leq t$ instead of a single parameter in a one-dimensional
parametric fluctuation model. Following a similar construction to Algorithm 4
of~\cite{Luedtke2017}, we can define a TMLE that is sequentially doubly robust:
Algorithm~\ref{alg:seq-tmle}.
This follows similarly to
Algorithm~\ref{alg:riesz-tmle}, but instead of a
one-dimensional parameter, $\varepsilon_t$ is now a \textit{function} estimated
non-parametrically via an arbitrary learning algorithm, and every time a
fluctuation is performed, one must update the $\varepsilon_t$, including for
all prior time steps too.

\begin{algorithm}
    \caption{Sequentially doubly-robust Riesz TMLE}
    \label{alg:seq-tmle}
    \begin{algorithmic}[1]
    \State Fit Riesz representers $\alpha_1, \ldots, \alpha_T$.
    \State Set $h_{T+1}(\bar{A}_{T+1}, \bar{L}_{T+1}; \overline{Q}_{T+1}^\star)
      = Y$.
    \For{$t = T, \ldots, 1$}
        \State {Initialize $\overline{Q}_{t}^{t+1}(\bar{A}_t, \bar{L}_t) =
          1/2$.}
        \For{$s = t,\ldots, 1$}
        \AlgStep{Fit nonparametric regression on $(\bar{A}_s, \bar{L}_s)$
          with outcome $h_{t+1}(\bar{A}_{t+1}, \bar{L}_{t+1};
          \overline{Q}_{t+1}^\star)$, appropriate link function, offset
          $\text{link}[\overline{Q}_{t}^{s+1}(\bar{A}_t, \bar{L}_t)]$, no
          intercept, and weights $\prod_{u=s+1}^t \alpha_u$.}
        \AlgStep{Update $\text{link}[\overline{Q}_{t}^s(\bar{A}_t,
          \bar{L}_t)] = \text{link}[\overline{Q}_{t}^{s+1}(\bar{A}_t,
          \bar{L}_t)] + \varepsilon_{n,t}^s(\bar{A}_s, \bar{L}_{s})$, where
          $\varepsilon_{n,t}^s(\bar{A}_s, \bar{L}_{s})$ denotes the function
          learned from the nonparametric regression.}
        \EndFor
        \State Set $\overline{Q}_{t}^\star =
          \overline{Q}_{t}^1$.
    \EndFor
    \State Compute the final estimate:
          \begin{equation*}
            \psi_{n}^{\star} = \frac{1}{n}\sum_{i=1}^nh_1(\bar{A}_{1},
            \bar{L}_{1}; \overline{Q}_1^\star) \ .
          \end{equation*}
    \end{algorithmic}
\end{algorithm}

Finally, we can consider a TMLE for parameters satisfying Theorem
\ref{thm:eif-twophase}, which we define formally in
Algorithm~\ref{alg:twostage-tmle}.
In this algorithm, the complete-data EIF is first computed on the coarsened
(observed) data, and is then learned as the outcome of a regression to predict
on the complete (uncoarsened) data. This regression is fluctuated using the
Riesz representer of the sampling structure as the clever covariate. This
method differs from prior literature~\citep{Hejazi2020, Rose2011}: instead of
targeting a particular nuisance, such as the conditional probability of
remaining observed despite the coarsening process, one instead tilts the
conditional mean to be agnostic to the form of the Riesz representer. Since
only the outer regression is targeted, this may not entirely satisfy the
definition of a TMLE, but similar procedures can be constructed depending on
desired plug-in constraints (for example, see Section 4.2 of~\cite{Qiu2026} for
a construction that also targets the inner regression).

\begin{algorithm}
    \caption{Two-phase sampling-adjusted Riesz TMLE}
    \label{alg:twostage-tmle}
    \begin{algorithmic}[1]
        \State Estimate the uncentered complete-data EIF
          $\phi^{\text{uc}}(\hat{\P}_n^X)(X)$ using phase-2 data.

        \State Fit the regression $\eta(\Delta, V) =
          \E[\phi^{\text{uc}}(\hat{\P}_n^X)(X) \mid \Delta, V ]$ and the Riesz
          representer $\alpha(\Delta, V)$ using phase-1 and phase-2 data.

        \State Using units in the phase-2 sample, fit a one-dimensional GLM
          $\eta_{\varepsilon_n}$ with outcome
          $\phi^{\text{uc}}(\hat{\P}_n^X)(X)$, appropriate link function, offset
          $\text{link}[\eta(\Delta, V)]$, covariate $
          \alpha(\Delta, V)$, and no intercept.

        \State Update $\eta^\star$ as $\text{link}[\eta^\star(\Delta, V)] =
        \text{link}[\eta(\Delta, V)] + \varepsilon_n \alpha(\Delta, V)$
        where $\varepsilon_n$ denotes the estimated coefficient from the fitted
        one-dimensional GLM.

        \State Compute the final TML estimate
        \begin{equation*}
          \psi_{n}^{\star} = \frac{1}{n}\sum_{i=1}^n
            \eta^\star(\Delta = 1, V_i) \ ,
        \end{equation*}
        where $\eta^\star(\Delta = 1, V_i)$ represents the prediction under
        the imputed data structure where all units are observed.
    \end{algorithmic}
\end{algorithm}

As TML estimators that solve the efficient score equation (approximately, for
Algorithm~\ref{alg:seq-tmle}), the estimators obtained from
Algorithms~\ref{alg:riesz-tmle},
\ref{alg:seq-tmle}, and
\ref{alg:twostage-tmle}
will be consistent and asymptotically normal and achieve the semi-parametric
efficiency bound for $\psi_0$ under the conditions summarized in
Remark~\ref{rem:asymptotics}.

\begin{remark}\label{rem:asymptotics}
  For an estimator $\psi_n^\star = \Psi(\hat{\P}_n^\star)$ solving
  $\E_{\P_n}\phi(\hat{\P}_n^\star)(O) = 0$, note that
  \begin{equation*}
    \psi_n^\star - \psi_0 = \E_{\P_n}\phi(\P_0)(O) +
    (\P_n - \P_0)[\phi(\hat{\P}_n^\star) - \phi(\P_0)] +
    R(\hat{\P}_n^\star, \P_0) \ ,
  \end{equation*}
  where, for the functional of
  Theorem~\ref{thm:eif-recursive}, the
  second-order remainder is
  \begin{equation*}
    R(\hat{\P}_n^\star, \P_0) = \sum_{t=1}^T \E_{\P_0}\Big[\big(\omega_{n,t} -
    \omega_{t}\big)\big(\widetilde{Q}_{n,t} - \overline{Q}_{n,t}^\star\big)
    (\bar{A}_t, \bar{L}_t)\Big] \ ,
  \end{equation*}
  with $\omega_t = \prod_{k=1}^t \alpha_k$, $\omega_{n,t}$ its estimate, and
  \begin{equation*}
    \widetilde{Q}_{n,t}(\bar{A}_t, \bar{L}_t) =
    \E_{\P_0}[h_{t+1}(\bar{A}_{t+1}, \bar{L}_{t+1};
    \overline{Q}_{n,t+1}^\star) \mid \bar{A}_t, \bar{L}_t] \ ,
  \end{equation*}
  the regression of the estimated pseudo-outcome. Hence, $\psi_n^\star$ is
  asymptotically linear with influence function $\phi(\P_0)$, and so efficient,
  provided that (i) the Riesz representers are bounded; (ii) the empirical
  process term is $o_\P(n^{-1/2})$, as holds when $\phi(\hat{\P}_n^\star)$ is
  consistent for $\phi(\P_0)$ in $L_2(\P_0)$-norm and either Donsker conditions
  hold or the nuisance functions are estimated with cross-fitting; and (iii)
  $R(\hat{\P}_n^\star, \P_0) = o_\P(n^{-1/2})$, as holds, for example, when
  $\lVert \omega_{n,t} - \omega_t \rVert \, \lVert \widetilde{Q}_{n,t} -
  \overline{Q}_{n,t}^\star \rVert = o_\P(n^{-1/2})$ for each $t$. The remainder
  vanishes if, at each $t$, either $\omega_{n,t} = \omega_t$ or
  $\overline{Q}_{n,t}^\star = \widetilde{Q}_{n,t}$; for $T = 1$, it reduces to
  the doubly robust form of
  Section~\ref{sec:linear-functional}.
  For Algorithm~\ref{alg:twostage-tmle},
  the complete-data EIF is estimated as well, so the remainder additionally
  involves the complete-data second-order remainder of
  $\phi^{\text{uc}}(\hat{\P}_n^X)$; see,
  e.g.,~\cite{Luedtke2017} and~\cite{Chernozhukov2022c} for details.
\end{remark}

\ijb{One may wonder when the TMLEs given by Algorithms
\ref{alg:riesz-tmle},
\ref{alg:seq-tmle}, and
\ref{alg:twostage-tmle}
differ from TMLEs derived
by classical methods. In common settings, the mathematical forms of these TMLEs
will often overlap precisely. If the Riesz representers are also derived
exactly and estimated in the same way as suggested by the classical setting,
then the results should be the same up to machine precision. However, another
key advantage of the Riesz representation framework is that $\alpha$ can often
be estimated directly by minimizing a special loss function, sometimes called
the Riesz loss---a technique termed \textit{Riesz regression}. This can help
overcome commonly recognized instability issues that arise when estimating, for
example, inverse probability weights or density ratios
\citep{chernozhukov2022riesznet, Lee2025, Williams2025}. In this case, one
may expect the Riesz TMLE to achieve better finite-sample performance than
alternatives.}

\ijb{In more exotic settings, the Riesz TMLE may differ from existing TMLEs.
This is because, in Algorithms~\ref{alg:riesz-tmle},
\ref{alg:seq-tmle},
and~\ref{alg:twostage-tmle},
the fluctuation step is performed only around the nuisance estimate $\eta$. If
an investigator chooses to fluctuate the estimate around another nuisance
parameter, then the TMLE will yield different estimates in finite samples. This
has been studied in the two-phase sampling literature, where others have
discussed using a fluctuation term for the conditional probability of two-phase
sampling~\citep{Rose2011, Hejazi2020}. More recent literature on estimation
under two-phase sampling designs has proposed TMLEs that align more closely
with Algorithm~\ref{alg:riesz-tmle}~\citep{Qiu2026,
Landsiedel2026} by fluctuating around only conditional means.}

%%%%%%%%%%%%%%%%%%%%%%%%%%%%%%%%%%%%%%%%%%%%%%%%%%%%%%%%%%%%%%%%%%%%%%%%%%%%%%%
\section{Examples}\label{sec:examples}

The EIF given by Riesz representation in Theorem~\ref{thm:eif-recursive}
encompasses a range of estimands commonly of interest, especially those arising
from causal inference, to which our TMLE procedures may be applied. In this
section, we show how this type of EIF appears in a variety of common
statistical estimation problems.

\subsection{Treatment-specific mean and quantile with
  censoring}\label{sec:treatment-specific-mean}

Consider observing $i = 1, \ldots, n$ data units $X = (L, A, Y)$, with $Y$
representing an outcome, $A$ a treatment of interest, and $L$ a set of
confounders. A common estimand in causal inference is the treatment-specific
mean:
\begin{equation*}
    \psi_0 = \E_{\P_0}[\E_{\P_0}(Y \mid A = a, L)]
\end{equation*}
\noindent Its EIF takes the form
\begin{equation*}
    \phi(\P)(X) = \frac{\1(A = a)}{\P(A = a \mid L)}\Big(Y - \E(Y \mid A, L)
    \Big) + \E(Y \mid A = a, L) - \psi \ ,
\end{equation*}
which can easily be seen to follow the pattern of
Corollary~\ref{cor:eif-condmean}
with nuisances
\begin{align}\label{eq:ipw-rr}
    \eta(X) = \E(Y \mid A, L) \,\,\,\, \text{and} \,\,\,\,
    \alpha(A, L) = \frac{\1(A = a)}{\P(A = a \mid L)} \ .
\end{align}

\noindent As a reminder, $h$ is the evaluation functional $\eta \mapsto \eta(a, L)$.

A Riesz representer is a Radon-Nikodym derivative, and the outer expectation
integrates over $(L, A)$, with the original measure $\P(A, L) = \P(A \mid L)
\P(L)$ and its ``replacement'' $\P^\star = \1(A = a) \P(L)$, as the
intervention setting $A = a$ is equivalent to replacing the probability measure
with the point mass $\1(A = a)$. This produces the inverse probability weight
Riesz representer of Equation~\ref{eq:ipw-rr}.

\ijb{This idea extends to handling censored outcomes under a missing-at-random
(MAR) assumption. Consider data units observed according to $O = (L, A, C, CY)$,
where $C \in \{0, 1\}$ denotes a censoring indicator (equal to 1 if the unit is
observed). We are interested in the parameter}
\begin{equation*}
    \psi_0 = \E_{\P_0}[\E_{\P_0}(Y \mid A = a, C =1, L)] \ ,
\end{equation*}
\ijb{which, under MAR, is equivalent to the complete-data parameter above. The
EIF of this parameter takes the form}
\begin{align*}
    \phi(\P)(O) =~& \frac{\1(A = a)}{\P(A = a \mid L)}\frac{\1(C = 1)}
      {\P(C = 1 \mid A, L)}\Big(Y - \E(Y \mid A, C = 1, L) \Big)\\ &+
      \E(Y \mid A = a,C = 1, L) - \psi \ ,
\end{align*}
\ijb{with nuisances}
\begin{align}\label{eq:ipw-rr2}
    \eta(O) &= \E(Y \mid A, C, L) \,\,\,\, \text{and} \nonumber \\
    \alpha(C, A, L) &= \frac{\1(A = a)}{\P(A = a \mid L)}\frac{\1(C = 1)}
      {\P(C = 1 \mid A = a, L)} \ .
\end{align}

\ijb{Here, $h$ is the evaluation functional $\eta \mapsto \eta(a, 1, L)$. The Riesz representer accounts for the censoring indicator $C$ analogous
to how the treatment $A$ is handled, that is, by replacing the observed
distribution $\P(C = 1 \mid A = a, L)\P(A = a \mid L)\P(L)$ with the
``intervened'' point mass distribution $\P^\star = \1(C = 1)\1(A = a) \P(L)$,
setting every unit to ``uncensored''.}

We can also use Theorem~\ref{thm:eif-basic} to obtain
results that extend beyond those obtained by previous literature; for example,
when $\eta$ is not a conditional expectation function. Suppose that, rather
than a treatment-specific mean, interest lies in a covariate-averaged
treatment-specific \textit{conditional quantile}, such as a conditional median,
\begin{equation*}
  \psi_{0}^{\tau} = \E[Q^\tau(a, 1, L)] \ ,
\end{equation*}
where $Q^\tau(a, c, l)$ denotes the $\tau$-quantile of $Y$ given $A = a$, $C =
c$, and $L = l$. This parameter differs from a quantile of the counterfactual
outcome distribution, as in the quantile treatment effects studied
by~\cite{firpo2007efficient}, \cite{diaz2017efficient}, and~\cite{Sun2021},
which is not a linear functional (see Section~\ref{sec:discussion}).

\noindent In this case, $\eta$ is the conditional $\tau$-quantile function
$Q^{\tau}(A, C, L)$. Following Appendix B of~\cite{hines2022demystifying},
its pathwise derivative is
\begin{equation*}
  \frac{\delta_{A, C, L}}{d\P(A, C, L)}\frac{\tau - \1(Y \leq Q^\tau(A, C,
  L))}{f(Q^\tau(A, C, L) \mid A, C, L)} \ ,
\end{equation*}
where $f(\cdot \mid A, C, L)$ denotes the conditional density of $Y$, so that
the EIF of $\psi^\tau = \E[h(O; Q^\tau)]$, with $h(O; Q^\tau) = Q^\tau(a, 1,
L)$, is
\begin{equation*}
  h(O; Q^\tau) - \psi^\tau + \alpha(C, A, L)\left(\frac{\tau -
  \1(Y \leq Q^\tau(A, C, L))}{f(Q^\tau(A, C, L) \mid A, C, L)}\right) \ ,
\end{equation*}
with nuisances
\begin{align}\label{eq:ipw-rr3}
    \eta(O) &= Q^\tau(A, C, L) \,\,\,\, \text{and} \nonumber \\
    \alpha(C, A, L) &= \frac{\1(A = a)}{\P(A = a \mid L)}\frac{\1(C = 1)}
      {\P(C = 1 \mid A = a, L)} \ .
\end{align}

Again, $h$ is the evaluation functional $\eta \mapsto \eta(a, 1, L)$, but now $\eta$ is the conditional quantile. The above has a form similar to the EIF of the treatment-specific mean, except
it multiplies the Riesz representer not by the residuals of a conditional mean
regression, but by a scaled derivative of the pinball loss optimized for
quantile regression. \ijb{The Riesz representer is identical to the mean case
because applying the treatment $A = a$ induces the same Radon-Nikodym
derivative.}

\subsection{Longitudinal Treatment Regimes}

Efficient estimation for the longitudinal causal inference problem with a
sequence of discrete treatments has been described by~\cite{vanderLaan2012},
and an approach for continuous treatments by~\cite{Diaz2021}. In the discrete
treatment regime, we can take as target parameter the counterfactual mean under
treatment regime $\bar{a}$. This estimand is a recursively defined sequence of
conditional means identified from the g-formula:
\begin{equation*}
    \psi_0 = \E[\E[\cdots \E[\E[Y \mid \bar{A}_T = \bar{a}_T, \bar{L}_T] \mid
      \bar{A}_{T-1} = \bar{a}_{T-1}, \bar{L}_{T-1}] \cdots \mid A_1 = a_1,
      \bar{L}_1]] \ .
\end{equation*}
Let $\overline{Q}_T(\bar{A}_T, \bar{L}_T) = \E(Y \mid \bar{A}_T, \bar{L}_T)$
and, for $t < T$, let a given step in the sequence of nuisance regressions be
$\overline{Q}_t(\bar{A}_t, \bar{L}_t) = \E[\overline{Q}_{t+1}(\bar{A}_t,
a_{t+1}, \bar{L}_{t+1}) \mid \bar{A}_t, \bar{L}_t]$, with $\overline{Q}_{T+1}
\coloneqq Y$. The EIF is given by $\phi(\P)(O) = \overline{Q}_1(a_1,
\bar{L}_1) - \psi + \sum_{t = 1}^T D_t$ where
\begin{equation}
    D_t = \frac{\1(\bar{A}_{t} = \bar{a}_{t})}{\prod_{k=1}^t g_k(a_k \mid
      \bar{A}_{k-1}, \bar{L}_{k})}\Big(\overline{Q}_{t+1}(\bar{a}_{t+1},
      \bar{L}_{t+1}) - \overline{Q}_t(\bar{A}_{t}, \bar{L}_{t})\Big) \ ,
\end{equation}
with $g_k(a_k \mid \bar{A}_{k-1}, \bar{L}_k) = \P(A_k = a_k \mid
\bar{A}_{k-1}, \bar{L}_k)$ the conditional probability of receiving treatment
$a_k$ at time $k$. In this estimand, $h_{t+1}$ is the evaluation functional $\overline{Q}_{t+1} \mapsto \overline{Q}_{t+1}(\bar{A}_t,
a_{t+1}, \bar{L}_{t+1})$.
From visual inspection, it is clear to see that this form
satisfies the form of Theorem~\ref{thm:eif-recursive}
with
\begin{align*}
    \alpha_t(\bar{A}_t, \bar{L}_t) &= \frac{\1(A_t = a_t)}{g_t(a_t \mid
      \bar{A}_{t-1}, \bar{L}_{t})} \ ,
\end{align*}
so that the weight in $D_t$ is $\omega_t = \prod_{k=1}^t \alpha_k$.
Here, each $\alpha_t(\cdot)$ represents the Radon-Nikodym derivative of the
counterfactual density with respect to the observed density. \ijb{The derivation
follows the same logic as Section \ref{sec:treatment-specific-mean}}.

One can also consider a \textit{modified treatment policy}. In this setting,
rather than a binary contrast between $A = 1$ and $A = 0$, we consider an
intervention that changes the natural value of $A_t$ to some ``intervened''
value $A_t^d$. \cite{Diaz2021} describe this setting, which is similar:
the sequential regressions are now $\overline{Q}_t(\bar{A}_t, \bar{L}_t) =
\E[\overline{Q}_{t+1}(\bar{A}_t, A_{t+1}^d, \bar{L}_{t+1}) \mid \bar{A}_t,
\bar{L}_t]$ for $t < T$, and the EIF is $\phi(\P)(O) =
\overline{Q}_1(A_1^d, \bar{L}_1) - \psi + \sum_{t = 1}^T D_t$ but now
\begin{equation}
    D_t = \prod_{k=1}^t \frac{g_k^d(A_k \mid \bar{A}_{k-1}, \bar{L}_k)}
      {g_k(A_k \mid \bar{A}_{k-1}, \bar{L}_k)}\Big(\overline{Q}_{t+1}(\bar{A}_t,
      A_{t+1}^d, \bar{L}_{t+1}) - \overline{Q}_t(\bar{A}_{t}, \bar{L}_{t})\Big)
      \ ,
\end{equation}
where $g_k^d(\cdot \mid \bar{A}_{k-1}, \bar{L}_k)$ represents the conditional
density of $A_k^d$ and $\overline{Q}_{T+1} \coloneqq Y$. In this estimand, $h_{t+1}$ is the evaluation functional $\overline{Q}_{t+1} \mapsto \overline{Q}_{t+1}(\bar{A}_t, A_{t+1}^d, \bar{L}_{t+1})$.
From visual
inspection, this is the same as the discrete case, but with a different Riesz
representer
\begin{equation*}
    \alpha_t(\bar{A}_t, \bar{L}_t) = \frac{g_t^d(A_t \mid \bar{A}_{t-1},
    \bar{L}_t)}{g_t(A_t \mid \bar{A}_{t-1}, \bar{L}_t)} \ .
\end{equation*}
\ijb{As a Radon-Nikodym derivative $d\P^\star / d\P$, this Riesz representer
can be derived by multiplying and dividing by the observed treatment
distribution so that each nested integral in the estimand integrates over the
observed distribution $d\P = g_t(A_t \mid \bar{A}_{t-1}, \bar{L}_t)$ instead of
the interventional distribution $d\P^\star = g_t^d(A_t \mid \bar{A}_{t-1},
\bar{L}_t)$ (the conditioning variables are common to the numerator and
denominator). Hence, despite being composed of complex sequences of nested
parameters, the EIFs of longitudinal estimands can easily be obtained even for
more complex treatment schemes.}

\subsection{Causal Mediation Analysis}\label{sec:mediation}

In causal mediation analysis, one typically works with the data unit $O = (L, A,
M, Y)$, with the goal of analyzing the effect of the treatment $A$ on the
outcome $Y$ through the effect of the mediating variable $M$. A common target
estimand is the \textit{natural direct effect} (NDE)~\citep{Robins1992,
pearl2001direct}, which represents the contrast
\begin{equation*}
    \psi_0 = \underbrace{\E_{\P_0}(\E_{\P_0}[\E_{\P_0}(Y \mid A = 1, M, L) \mid
      A = 0, L])}_{\theta} -
      \underbrace{\E_{\P_0}(\E_{\P_0}(Y \mid A = 0,
      L))}_{\text{Counterfactual mean}}
\end{equation*}

For brevity we will consider the M-functional $\theta$, as the additional
contribution to the EIF from the counterfactual mean is already known. The EIF
of $\theta$ is given by \cite{TchetgenTchetgen2012} as
\begin{align*}
    \phi_{\theta}(O) =~& \frac{\1(A=1)}{\P(A=1 \mid L)}
    \frac{\P(M \mid A = 0, L)}{\P(M \mid A = 1, L)}
    \Big(Y - \E(Y \mid A = 1, M, L)\Big)\\
    &+ \frac{\1(A=0)}{\P(A=0 \mid L)}
      (\E(Y \mid A = 1, M, L) - \eta(1, 0, L)) \\
    &+ \eta(1, 0, L) - \theta \ ,
\end{align*}
which uses the nuisance parameter shorthand
\begin{equation*}
    \eta(1, 0, L) = \E[\E(Y \mid  A = 1, M, L) \mid A = 0, L] \ .
\end{equation*}
\noindent
\ijb{
This EIF does not follow the form of Theorem~\ref{thm:eif-recursive},
because the inner conditional map $f \mapsto \E_{\P_0}[f(1, M, L) \mid A, L]$
is not a bounded functional: it evaluates $f$ at $A = 1$ while conditioning on
$A$, so it is unbounded on $L_2(\P_0(\cdot \mid A, L))$ on the event $\{A =
0\}$. It does, however, take the form of Remark~\ref{rem:product},} with two
nested
nuisances,
\begin{align*}
  \overline{Q}_2(A, M, L) &= \E(Y \mid A, M, L) \quad \text{and} \\
  \overline{Q}_1(A, L) &= \E[\overline{Q}_2(1, M, L) \mid A, L] = \eta(1, A,
  L) \ ,
\end{align*}
so that $h_2 \coloneqq \overline{Q}_2 \mapsto\overline{Q}_2(1, M, L)$ and $h_1 \coloneqq \overline{Q}_1 \mapsto\overline{Q}_1(0, L)$,
and with cumulative Riesz representers $\omega_t$, each following from a
change of measure, given by
\begin{align*}
  \omega_1(A, L) = \frac{\1(A=0)}{\P(A=0 \mid L)} \,\,\,\, \text{and} \,\,\,\,
    \omega_2(A, M, L) = \frac{\1(A=1)}{\P(A=1 \mid L)}
      \frac{\P(M \mid A = 0, L)}{\P(M \mid A = 1, L)} \ .
\end{align*}
\ijb{Hence, $\theta$ can be estimated with Algorithm~\ref{alg:riesz-tmle}
using direct estimates of $\omega_1$ and $\omega_2$, though not with
Algorithm~\ref{alg:seq-tmle}.} A similar strategy can
be employed when mediation is expanded to other settings. For example,
Theorem 1 of~\cite{Zheng2017} provides the interventional direct effect's EIF
for a longitudinal data structure, which can be cast in the form of our
Theorem~\ref{thm:eif-recursive} using the same
logic reviewed above. However, one must be careful, as some mediational
functionals, such as those defined under stochastic
interventions~\citep{Hejazi2022} may not be directly representable as sequences
of conditional expectations, necessitating a more careful application of
Theorem~\ref{thm:eif-basic} instead.

\section{Simulation Study: Estimator Properties}\label{sec:simulation}

We implemented our proposed Riesz TMLE algorithms along with asymptotically
equivalent one-step estimators in the \texttt{R} package \texttt{\{RieszCML\}}
(short for Riesz Causal Machine Learning), available at
\url{https://github.com/nshlab/RieszCML}. The package supports one-step and
TML estimation procedures, along with flexible nuisance function estimation via
the super learner algorithm~\citep{vdl2007super}.
% Our sequential Riesz TMLE algorithm and Theorem 2 naturally support our
% package's ability for estimation in both the cross-sectional and longitudinal
% setting, as well as supporting estimation of novel ``composed'' parameters.
A central design feature of the \texttt{\{RieszCML\}} package is the separation
of (i) the definition of the target parameter through a Riesz representer and
(ii) the estimation of nuisance components. In the following sections, we
present simulation experiments demonstrating comparable performance of our
proposed procedures to their previously described, more standard counterparts.
%The basic architecture of the \texttt{\{RieszCML\}} is illustrated in the
%example below. Riesz representers are defined as one-sided formula expressions
%that may use terms from the data and nuisance functions.

%\lstinputlisting{intro_to_RieszCML_listing1.R}

\subsection{Data-generating process}

To set up our simulation experiments, we consider a single time point setting
with a binary treatment $A$, covariates $L = (L_1, L_2, L_3, L_4, L_5)$, and a
binary outcome $Y$. The target parameter is the ATE, $\psi_0 =
\E[\overline{Q}(1, L) - \overline{Q}(0, L)]$, the risk difference between
treatment and control. Covariates are generated as
$$
L_1, L_2, L_3 \overset{i.i.d.}{\sim} \mathcal{N}(0, 1),
\quad  L_4 \overset{i.i.d.}{\sim} \mathrm{Bernoulli}(0.5),
\quad L_5 \overset{i.i.d.}{\sim} \mathrm{Uniform}(-1,1) \ .
$$
Treatment assignment follows a logistic model depending on main terms and
select interactions:
$$
  \P(A = 1 \mid L) =
    \mathrm{expit}(-0.4 + 0.6 L_2 - 0.5 L_3 + 0.5 L_4 L_5 - 0.4 L_1 L_2) \ ,
$$
and the outcome regression is given by
$$
  \P(Y = 1 \mid A = a, L) = \mathrm{expit}(-0.8 + 0.9 a + 0.6 L_2 +
    0.8 L_4 L_5 + 0.7 a L_1 - 0.6 a L_2) \ .
$$
% This data-generating process deliberately includes interactions and omits
% some covariates so that flexible machine learning methods such as the Super
% Learner algorithm are beneficial. % TODO add citation
The true value of $\psi_0$ and semi-parametric efficiency bound were
approximated via Monte Carlo integration with $2 \times 10^6$ draws from the
data-generating process.

\subsection{Comparison of methods}

We compared our proposed Riesz-TMLE to the standard TMLE method as implemented
in the \texttt{\{tmle3\}} package. Both applied a logistic fluctuation to the
initial outcome regression estimator.

We considered sample sizes ranging from $n = 100$ to $6400$. For each sample
size, we performed 2,000 Monte Carlo replications. Both methods were
implemented using super learning for the nuisance parameters. The
\texttt{\{simChef\}} \texttt{R} package provided the architecture for this
simulation study~\citep{duncan2024simchef}. For each estimator, we calculated
estimated bias, coverage of 95\% Wald-style confidence intervals, mean squared
error (MSE), and asymptotic relative efficiency ($n$ times the MSE, divided by
the efficiency bound). We summarize the results in
Figure~\ref{fig:convergence_diagnostics}.

\begin{figure}[H]
    \centering
    \includegraphics[width=0.95\linewidth]{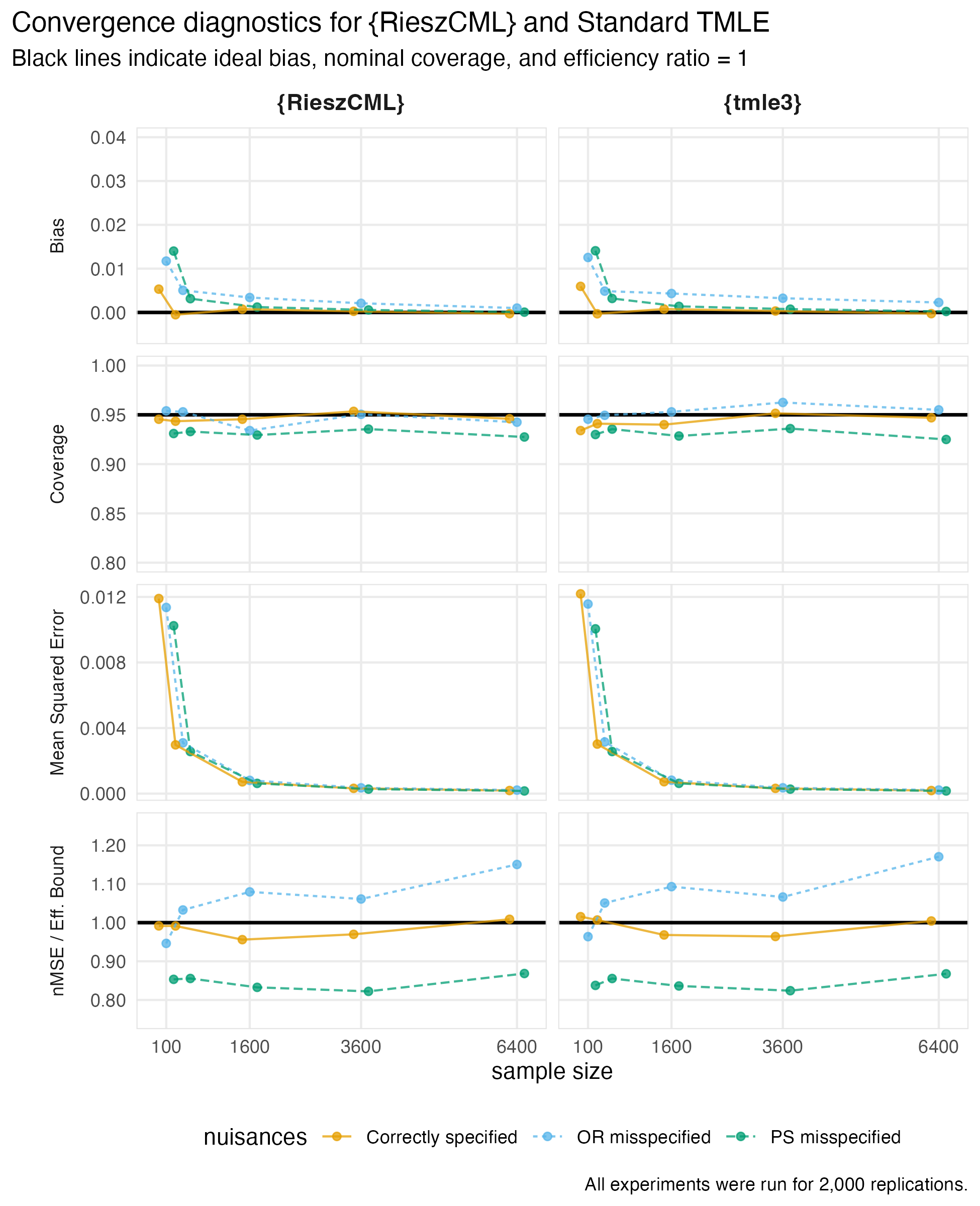}
    \caption{Convergence diagnostics for estimators of the ATE functional.
      Panels display bias, coverage, mean squared error (MSE) and $n \times
      \mathrm{MSE} / \mathrm{Efficiency\ Bound}$. Results are shown separately
      for the \texttt{RieszCML} and \texttt{tmle3} frameworks, and for three
      nuisance configurations: all correctly specified, outcome regression
      misspecified, and propensity score misspecified.}
    \label{fig:convergence_diagnostics}
\end{figure}

For the Riesz-type estimators, nuisance functions were estimated using
\texttt{nadir::super\_learner()}~\citep{testaNadirSuperLearning2026}, while for
\texttt{tmle3} we use \texttt{sl3::Lrnr\_sl}~\citep{coyle2020sl3}. In both
cases, the super learner library consisted of the \texttt{mean} learner (which
predicts using the observed outcome mean), GLM, MARS, and $\texttt{glmnet}$,
combined together in a weighted ensemble using a non-negative least-squares
metalearner.

To assess double robustness, we considered three nuisance specifications.
\textit{Correctly specified}: both $\overline{Q}_0(a, L)$ and $g_0(L)$ are
estimated using a flexible super learner with a library of learners that can
approximate the true data-generating process; \textit{outcome regression
misspecified}: $\overline{Q}_0(a, L)$ is replaced with a constant mean
estimator, while $g_0(L)$ is estimated flexibly as in the correctly specified
case; and \textit{propensity score misspecified}: $g_0(L)$ is replaced with a
constant mean estimator, while $\overline{Q}_0(a, L)$ is estimated flexibly as
in the correctly specified case.

% These results allow us to verify that the Riesz TMLE (like the standard TMLE)
% retains consistency whenever either nuisance component is correctly specified.
The Riesz-TMLE and standard TMLE exhibit nearly identical performance across
all settings. In particular, all estimators show small bias that decreases with
sample size. Coverage is close to the nominal 95\% level across all sample
sizes and nuisance specifications, falling slightly below it (to roughly 93\%)
when the propensity score is misspecified. When both nuisances are correctly
specified, $n \times \mathrm{MSE}$ of both TMLE estimators closely tracks the
semi-parametric efficiency bound. When either nuisance is misspecified, the
estimators remain consistent but are no longer expected to be efficient, so
their relative efficiency need not equal one: $n \times \mathrm{MSE}$ exceeds
the bound when the outcome regression is misspecified and falls below it when
the propensity score is misspecified.

% Across all experiments, the Riesz TMLE performs comparably to the standard
% TMLE, demonstrating that the Riesz representation based construction yields a
% valid and efficient estimator in practice. Overall, the simulation study
% supports that the Riesz TMLE attains the same semi-parametric efficiency bound
% as the standard TMLE.

%%%%%%%%%%%%%%%%%%%%%%%%%%%%%%%%%%%%%%%%%%%%%%%%%%%%%%%%%%%%%%%%%%%%%%%%%%%%%%%
\section{Applications in Real and Synthetic Data}\label{sec:data-analysis}

\subsection{Counterfactual infection risk with two-phase sampling and censored
outcomes}

To illustrate how our proposed approach can simplify the implementation of
efficient estimation strategies, we re-examine a prior analysis reported
by~\cite{Hejazi2020} of data from the HIV Vaccine Trials Network's (HVTN) 505
trial~\citep{Hammer2013}. HVTN 505 (NCT00865566) was a phase 2b randomized,
placebo-controlled trial of an investigational vaccine regimen that enrolled
2504 HIV-negative participants, who were randomized one-to-one to receive an
active vaccine (DNA/rAd5) or placebo. A secondary aim of this study was to
investigate whether and the degree to which vaccine-induced immune responses
were associated with HIV infection by the end of primary follow-up (month 24).

Employing a two-phase case-control sampling design~\citep{Janes2017}, candidate
immune responses were measured for all HIV cases diagnosed between week 28 and
month 24 and five uninfected controls matched to each case based on BMI and
race/ethnicity. \cite{Janes2017} and \cite{fong2018modification} analyzed
immune responses derived from blood drawn at the week 26 visit, both finding
CD4+ and CD8+ polyfunctionality scores to be associated with risk of HIV
infection status by month 24. \cite{Hejazi2020} developed new methodology based
on the modified treatment policy framework of~\cite{haneuse2013estimation} to
estimate the counterfactual risk of HIV infection under investigator-supplied
shifts in standardized polyfunctionality scores of the CD4+ and CD8+ immune
markers, adapting and expanding EIF-based corrections for two-phase sampling
previously introduced by~\cite{Rose2011}; their methods were implemented in the
\texttt{\{txshift\}} \texttt{R} package~\citep{hejazi2020txshift} and used to
define a novel vaccine efficacy measure, \textit{stochastic-interventional
vaccine efficacy}~\citep{hejazi2023stochastic, gilbert2024four}, applied for
immune correlates analyses of CoVPN's vaccine efficacy
trials~\citep{huang2023stochastic, hejazi2023stochastic}. The findings of this
analysis generally echoed those of~\cite{Janes2017}
and~\cite{fong2018modification}.

We reconsider the previously reported analysis of~\cite{Hejazi2020}, applying
our Riesz-based TMLE procedure, presented in Algorithm~\ref{alg:twostage-tmle},
to develop efficient estimators of the counterfactual mean of HIV infection
risk under hypothetical modifications of CD4+ and CD8+ polyfunctionality
scores, doing so as an application of
Theorem~\ref{thm:eif-twophase}. In this study context, the observed
data unit is $O = (L,\Delta S, Y, \Delta)$, where $\Delta$ is the indicator of
selection into the phase-two sample and follows
\begin{equation*}
    \P(\Delta = 1 \mid Y, L) =
    \begin{cases}
      1, & Y = 1\\
      \pi(L), & \text{otherwise} \ ,
    \end{cases}
\end{equation*}
where $0 < \pi(L) < 1$ is the two-phase sampling probability for uninfected
controls, selected by~\cite{Janes2017} by a matching procedure. Our target
parameter is $\E[\P(Y = 1 \mid S + \delta, L)]$, the population-level risk of
HIV infection among vaccine recipients when standardized CD4+ or CD8+
polyfunctionality score $S$ is modified from its natural value $S$ based on a
modified treatment policy to yield $S + \delta$. Applying
Theorem~\ref{thm:eif-twophase}, we obtain the EIF as
\begin{align*}
    \phi(\P)(O) =& \,\, \E\left(\E[Y \mid S + \delta, L] + \alpha_2(S, L)
      [Y - \E(Y \mid S, L)] \mid \Delta = 1, Y, L\right) - \psi_0 \\
    &+ \alpha_1(\Delta, Y, L) \E[Y \mid S + \delta, L] - \alpha_1(\Delta, Y, L)
      \E\left(\E[Y \mid S + \delta, L] \mid \Delta, Y, L\right) \\
    &- \alpha_1(\Delta, Y, L) \E\left(\alpha_2(S, L)[Y - \E(Y \mid S, L)] \mid
      \Delta, Y, L\right)\\
    &+ \alpha_1(\Delta, Y, L)\alpha_2(S, L)[Y - \E(Y \mid S, L)] \ ,
\end{align*}
where the Riesz representers are given by
\begin{align*}
    \alpha_2(S, L) &= \frac{d\P(S - \delta \mid L)}{d\P(S \mid L)} &
    \alpha_1(\Delta, Y, L) &= \frac{\1(\Delta = 1)}{\P(\Delta = 1 \mid Y, L)}
\end{align*}
While the form of the EIF appears involved, it is the EIF given
by~\cite{Hejazi2020}, written in terms of Riesz representers and regression
nuisance functions; this representation facilitates the construction of
estimation procedures simplified by using Algorithm~\ref{alg:twostage-tmle}.

In our re-analysis, following~\cite{Hejazi2020}, we estimated the mean risk of
HIV infection across a range of candidate modifications to CD4+ and CD8+
standardized polyfunctionality scores. Estimates were produced using the TMLE
procedure outlined in
Algorithm~\ref{alg:twostage-tmle};
see
Figure~\ref{fig:data-analysis}.
Throughout, to aid comparability, we re-use nuisance estimates computed
by~\cite{Hejazi2020}, including two-phase sampling probabilities and
conditional densities of polyfunctionality scores needed for the construction
of first- and second-stage Riesz representers. When it is necessary to estimate
nuisance functions, we use super learning~\citep{vdl2007super}; our weighted
ensemble of learning algorithms included generalized additive models, random
forest, gradient boosted trees, penalized and unpenalized GLMs with all
first-order interaction terms, and two variations of highly adaptive
lasso~\citep{van2017generally, hejazi2020hal9001}.

\begin{figure}[h]
    \centering
    \includegraphics[width=0.9\linewidth]{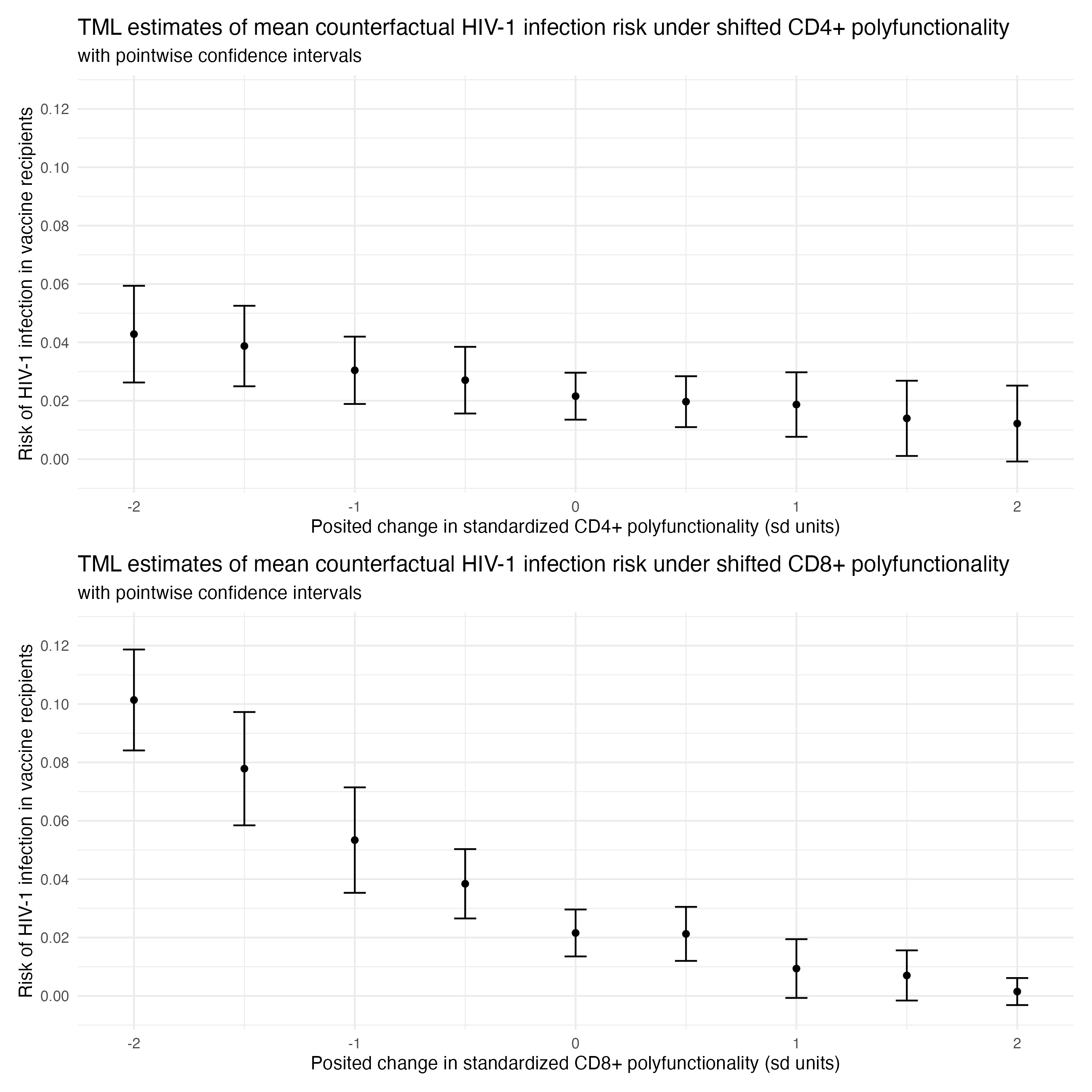}
    \caption{TML estimates of the counterfactual mean of HIV-1 infection in
    vaccine recipients under additive shifts in CD4+ (top) and CD8+ (bottom)
    standardized polyfunctionality scores, including pointwise Wald-style
    confidence intervals. Analogous estimates appear in Figure 2
    of~\cite{Hejazi2020}; see the text for a comparison.}
    \label{fig:data-analysis}
\end{figure}

Our analysis reaches the same conclusion as that previously reported: that
hypothetical reductions in CD4+/CD8+ polyfunctionality scores would increase
the risk of HIV infection; the two analyses also agree at the observed
polyfunctionality scores ($\delta = 0$). For hypothetical increases to
polyfunctionality scores, our estimates decrease, whereas those
of~\cite{Hejazi2020} lie above their estimate at $\delta = 0$. Positive shifts
move many scores beyond the range observed in the trial~\citep{Hejazi2020}, so
estimates under such shifts rely on extrapolation; the discrepancy may reflect
this, together with the difference in targeting strategies noted in
Section~\ref{sec:tmle}, in which \cite{Hejazi2020}
construct a TMLE by based on fluctuation of the two-phase sampling
probabilities.

% It does deviate slightly from the weak upward trend in HIV risk at higher
% polyfunctionality scores observed by , but we posit that our method, by
% targeting a conditional expectation instead of a conditional sampling
% probability, actually serves to more effectively correct possible
% instabilities
% in the estimates of this previous work.

\subsection{Stochastic-interventional vaccine efficacy with two-phase sampling
and censored outcomes}

Here, we demonstrate the construction of, estimation of, and inference for a
recently proposed target parameter that our Sequential Riesz EIF
(Theorem~\ref{thm:eif-recursive}) and
software package \texttt{\{RieszCML\}} render more readily implementable, the
stochastic-interventional vaccine efficacy (SVE) effect measure introduced
by~\cite{hejazi2023stochastic} and applied to assessing correlates of
protection in the COVID-19 Prevention Network's (CoVPN) COVE trial, a
phase~3 efficacy trial of the mRNA-1273 vaccine
(NCT04470427)~\citep{gilbert2024four, huang2023stochastic}. To the best of
our knowledge, no implementation of estimation procedures for the SVE is
available in open-source software. The SVE considers a contrast of a joint
intervention that, for vaccine recipients, applies a modified treatment
policy~\citep{haneuse2013estimation, Hejazi2020} to a candidate correlate of
protection, typically a post-vaccination immune marker $S$ to shift each
vaccine recipient's marker from its natural value $S$ to $S + \delta$; this is
contrasted against the placebo arm's risk.

We present a simulation study that uses a data-generating process calibrated to
mimic salient features of the CoVPN COVE trial, in which we estimate a
$\mathrm{SVE}(\delta)$ across a range of candidate values $\delta$ to
demonstrate the implementation of our theoretical results in the
\texttt{\{RieszCML\}} package. We stress that the results and figures presented
here should be interpreted as a validation study and demonstration of the
software and our theory (e.g., Theorems~\ref{thm:eif-basic}
and~\ref{thm:eif-recursive})
and not as related to the findings of the CoVPN COVE trial.

Key features of the data-generating process include that the outcome is subject
to censoring and, as standard in vaccine efficacy trials, that the immune marker
is measured via two-phase sampling~\citep{follmann2006augmented}. Here, the
observed data is $n$ i.i.d.~copies of the data unit $O = (L, A, RS, C, CY, R)$,
where $L$ are baseline covariates, $R$ is the indicator of two-phase sampling
(immune markers are observed in a breakthrough case-enriched manner), $A$ is
the indicator of randomized assignment to vaccine or placebo, $S$ is the
candidate immune correlate (e.g., neutralizing antibody titer at Day 57
post-vaccination in COVE), $C$ is the indicator of remaining uncensored (so
that $Y$ is observed when $C = 1$), and $Y$ is the outcome. We write
$\psi_1(\delta) = \E_{S, L \mid A = 1}[\E(Y \mid S + \delta, L, A = 1, C = 1)]$
for the vaccine arm risk under the hypothetically modified candidate immune
correlate $S + \delta$ and $\psi_0$ for the placebo arm risk. The target
parameter is defined as $\mathrm{SVE}(\delta) = 1 - \psi_1(\delta) / \psi_0$.
The SVE quantifies how much higher or lower vaccine efficacy would have been if
the vaccine's hypothesized mechanisms of action through the given immune marker
were stronger ($\delta > 0$) or weaker ($\delta < 0$) than that observed
($\delta = 0$). Following the CoVPN COVE
trial~\citep{gilbertImmuneCorrelatesAnalysis2022}, we calibrated our
data-generating process to generate $n = 28{,}000$ study participants
randomized one-to-one to vaccine or placebo, with a rare vaccine-arm outcome
and a placebo-arm risk of $6\%$. We also incorporated a two-phase sampling
design to measure candidate correlates in all observed vaccine-arm cases
alongside a $7\%$ control subcohort (1,100 phase-two vaccinees); the marker was
analyzed in standardized (SD) units.

\begin{figure}[htbp]
\includegraphics[width=\textwidth]{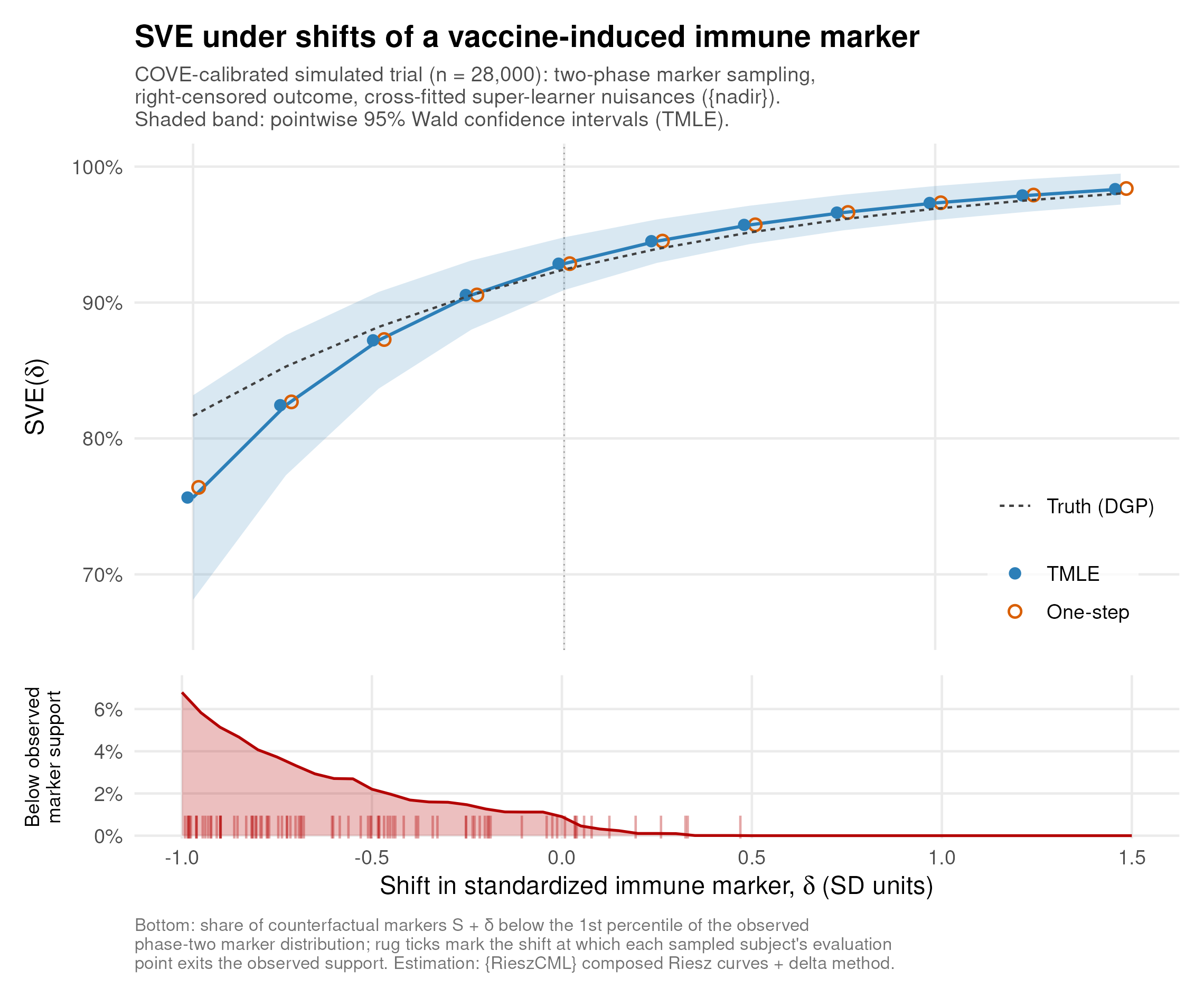}
\caption{Simulation Study. Using our \texttt{\{RieszCML\}} software
we estimated a stochastic vaccine efficacy $\mathrm{SVE}(\delta)$ parameter for
$\delta \in [-1, 1.5]$ using one-step and TML estimators based on
Theorem~\ref{thm:eif-recursive} and a
data-generating process calibrated to the CoVPN COVE trial.}
\label{fig:sve-cove}
\end{figure}

Using \texttt{\{RieszCML\}}, this estimand requires no new theoretical
derivation whatsoever: each complication contributes one stage of a composed
linear functional, where the Riesz representers and other terms are already
well known in the causal inference literature. The numerator $\psi_1(\delta)$
is a three-stage nested functional of the form of Theorem
\ref{thm:eif-recursive}. Indexing the stages
as in Theorem~\ref{thm:eif-recursive}, the
outermost stage uses the two-phase sampling representer $\alpha_1 = \1(R =
1)/\pi(V)$, with $\pi$ the sampling probability and $V = (L, C, CY)$ the
variables always measured in phase one. The middle stage carries the
representer for the modified treatment policy, $\alpha_2 = g(S - \delta \mid
L)/g(S \mid L)$, applied to the regression of the innermost stage on $(S, L)$
and evaluated at $S + \delta$. The innermost stage, the regression $\E(Y \mid
S, L, C = 1, A = 1)$, uses the inverse-probability-of-censoring weights
$\alpha_3 = \1(C = 1)/\P(C = 1 \mid S, L)$. Following the COVE correlates
analyses~\citep{gilbertImmuneCorrelatesAnalysis2022, gilbert2024four}, this
builds on the inverse-probability-of-sampling-weighted complete-data EIF.
Algorithm~\ref{alg:riesz-tmle} then yields a
procedure for constructing a TMLE for this parameter, and we additionally
implement an analogous one-step estimator for comparison. The clever covariates
needed for the TMLE algorithm in the second and outermost stages' fluctuation
steps thus are calculated automatically through the constructed product terms
$\omega_t$ (see Algorithm~\ref{alg:riesz-tmle})
rather than users ever having to specify or calculate the product clever
covariates of these stages themselves. The denominator $\psi_0$ is simpler,
being just a risk under no censoring among the placebo arm implemented in the
\texttt{\{RieszCML\}} package as a single object of the \texttt{RieszCurve}
class. Nuisance functions were estimated via crossfit super learning, where all
predictions were made out of fold. Since $\mathrm{SVE}(\delta) = 1 -
\psi_1(\delta) / \psi_0$ is a smooth transformation of linear functionals,
inference for $\mathrm{SVE}(\delta)$ follows from the delta method as described
in Section~\ref{sec:eif}
and use of the provided \texttt{riesz\_delta()} helper function.

% TABLE FOR SIMULATION  RESULTS 
\begin{table}

\caption{Simulation Study. Monte Carlo operating characteristics for
SVE($\delta$) estimators with the COVE-calibrated data-generating process. Each
scenario was run 2000 times.}
\label{tbl:sve-cove}
\centering
\begin{tabular}[t]{rlcccccc}
\toprule
$\delta$ & Estimator & Truth & Bias & Emp.\ SD & Mean SE & RMSE & Cover (\%)\\
\midrule
+1.50 & One-step & 0.980 & 0.0006 & 0.0134 & 0.0080 & 0.0134 & 87.3\\
+1.50 & TMLE & 0.980 & -0.0009 & 0.0085 & 0.0080 & 0.0086 & 91.5\\
\addlinespace
+1.00 & One-step & 0.969 & -0.0000 & 0.0075 & 0.0074 & 0.0075 & 94.0\\
+1.00 & TMLE & 0.969 & -0.0008 & 0.0070 & 0.0074 & 0.0070 & 95.3\\
\addlinespace
+0.50 & One-step & 0.952 & -0.0004 & 0.0075 & 0.0078 & 0.0075 & 95.7\\
+0.50 & TMLE & 0.952 & -0.0006 & 0.0074 & 0.0078 & 0.0075 & 96.0\\
\addlinespace
+0.00 & One-step & 0.924 & -0.0003 & 0.0100 & 0.0100 & 0.0101 & 94.3\\
+0.00 & TMLE & 0.924 & -0.0003 & 0.0100 & 0.0100 & 0.0100 & 94.3\\
\addlinespace
-0.50 & One-step & 0.882 & 0.0004 & 0.0193 & 0.0172 & 0.0193 & 91.0\\
-0.50 & TMLE & 0.882 & 0.0001 & 0.0190 & 0.0172 & 0.0190 & 91.1\\
\addlinespace
-1.00 & One-step & 0.817 & 0.0015 & 0.0545 & 0.0393 & 0.0545 & 84.7\\
-1.00 & TMLE & 0.817 & -0.0015 & 0.0496 & 0.0393 & 0.0496 & 83.5\\
\bottomrule
\end{tabular}
\end{table}

We present simulation results (Figure
\ref{fig:sve-cove})
showing that the \texttt{\{RieszCML\}} one-step and TML estimation procedures
recover $\mathrm{SVE}(\delta)$ across a grid of shifts $\delta \in [-1, 1.5]$
under two-phase sampling of a candidate correlate of protection and with
censored outcomes, with Wald-style asymptotic confidence intervals. When a
modified treatment policy shifts $S$ by a great enough magnitude to induce
positivity violations, the quality of the implemented estimators degrades,
especially for shifts of $\delta\leq-1$. Table~\ref{tbl:sve-cove}
reflects the performance of our estimators in a Monte Carlo simulation study,
reporting bias, RMSE, and coverage of these estimators across a range of
$\delta \in [-1, 1.5]$ where at extremal ends of the range, practical
positivity violations become a concern. Our TML estimators attain the nominal
95\% coverage for deltas in the middle range, but both estimators show that
they suffer from deteriorating coverage as a result of practical positivity
violations towards the extreme ends of the intervention range. For the range
where $\delta < 0$, this corresponds to when a substantial share of
counterfactual marker values $S + \delta$ fall below the support of the
observed phase-two marker distribution, as shown in our diagnostic plot in the
bottom half of Figure~\ref{fig:sve-cove}.

The lower panel of Figure~\ref{fig:sve-cove}
depicts a diagnostic for when the modified treatment policy on $S$ is large
enough to produce practical positivity violations, with estimates shown by rug
ticks of the greatest $\delta$ such that $S+\delta$ is below the empirical
$1$\textsuperscript{st} percentile of the observed support for each vaccinee; a
similar diagnostic could be constructed for when $S$ values are pushed above
the $99$\textsuperscript{th} percentile of observed support. Here, we focus
on the positivity violations that might explain why the truth appears towards
the upper edge of the estimated CI for $\delta \approx -1.0$ for the TMLE. Each
vaccine recipient in the phase-two marker group contributes one rug tick.
The shifts indicated in the rug ticks represent when the intervention magnitude
is large enough that each two-phase sampled individual's marker leaves the
empirical support of the data, corresponding to $g(S-\delta \mid L)/g(S
\mid L)$ being large enough to induce practical positivity violations. Rug
ticks of positive $\delta$ values reflect that the $S+\delta$ value is already
below the empirical 1\textsuperscript{st} percentile of $S$ for the $\delta$
used. The shaded curve shows the inverse-probability-of-sampling-weighted
estimate of
\begin{equation*}
\P \Bigl( S + \delta < q_{0.01} \Bigm| A = 1 \Bigr) =
  \E\left[ \frac{g(S-\delta \mid L)}{g(S \mid L)} \1\left\{S < q_{0.01}\right\}
  \Bigm\vert A = 1 \right] \ ,
\end{equation*}
where $q_{0.01}$ is the empirical 1\textsuperscript{st} percentile of $S$ among
the phase-two sample.

While estimation of the effects of stochastic \textit{treatment} regimes under
two-phase sampling and right-censoring has previously been possible via the
purpose-built \texttt{\{txshift\}} package~\citep{hejazi2020txshift}, to our
knowledge no general purpose targeted learning software has made the SVE
estimable without requiring further accompanying estimand-specific derivations
and computer implementation. We further note that our package makes possible
implementation
of a much wider class of estimands than $\mathrm{SVE}(\delta)$. Our package
provides the ability to perform targeted learning on nested sequential
estimands with Riesz-representable EIFs, making TMLE more readily feasible for
parameters that can be characterized as being composed of linear functionals of
nuisances.

%%%%%%%%%%%%%%%%%%%%%%%%%%%%%%%%%%%%%%%%%%%%%%%%%%%%%%%%%%%%%%%%%%%%%%%%%%%%%%%
\section{Discussion}\label{sec:discussion}

We introduced a framework for constructing asymptotically efficient estimators
of statistical functionals that can be represented as a sequence of one or more
linear functionals of nuisance functions, and which therefore share a common
efficient influence function (EIF) structure. The underlying representation of
the EIF permits construction of targeted minimum loss-based estimators (TMLEs)
based on a sequential regression formulation. Our proposed work unifies
efficient estimation strategies across a variety of diverse settings, including
causal mediation analysis, means of conditional quantiles, longitudinal settings
with time-varying treatment subject to treatment-confounder feedback, and
corrections for the use of two-phase sampling designs.

% \begin{table}[h!]\label{tab:compare}
% \centering
% \small
% \begin{tabularx}{\textwidth}{>{\raggedright\arraybackslash}p{2.8cm} X X}
% \toprule
%  & \textbf{Classical Methods} & \textbf{Riesz-Based Approach} \\
% \midrule
% \textbf{Conceptual Simplicity} 
%   & Generally requires solving a new calculus problem for each new estimand 
%   & Same form generalizes across all estimands in the class \\
% \addlinespace
% \textbf{Computational Automation} 
%   & Various software packages for different estimands, requiring major
%   maintenance
%   & Same algorithm for each estimand, bespoke packages not required \\
% \addlinespace
% \textbf{Scope of Applicability} 
%   & All pathwise differentiable parameter
%   & Bounded, linear (``Riesz-representable'') parameters, a subset of
%   pathwise-differentiable parameters \\
% \bottomrule
% \end{tabularx}
% \caption{Comparision of Riesz-based TMLE to classical methods}
% \end{table}

\ijb{A limitation of this work is that the target parameter is required to be a
\textit{linear} functional of a nuisance parameter, or a smooth function
of such parameters. While this covers a large class of mean-based parameters
commonly of interest in causal inference and other areas of statistics, it
excludes parameters such as the variance of a conditional average treatment
effect (a quadratic functional of the outcome regression), as well as
quantiles (though it does cover means of conditional quantiles, as discussed
in Section \ref{sec:treatment-specific-mean}). In addition,
estimands defined by nonlinear transformations of nuisance functions, such as
the incremental propensity score intervention~\citep{Kennedy2018}, pose a
challenge to estimation within our framework. Some related works propose the
use of asymptotic expansions to linearize functionals in order to apply Riesz
representation theory~\citep{vdl2025}, which may represent a relevant area of
further research.}

A practical advantage of our theoretical framework is that it facilitates the
development of open-source statistical software that remains applicable across
a wide variety of problem settings. We implemented our proposed TMLE algorithm
for estimands with Riesz-based EIFs in the \texttt{\{RieszCML\}} \texttt{R}
package. An avenue of future investigation is the continued development of
tools that simplify the implementation of efficient estimators. A second
advantage of our focus on Riesz-based EIFs is that the Riesz representer
need not be constructed from estimates of its components (e.g., a propensity
score): \cite{Chernozhukov2022} suggest estimating the Riesz representer
directly, or ``automatically,'' using knowledge of $h$, by optimizing a
tailored loss function---an approach termed \textit{Riesz regression}---which
can result in more stable estimates when the Riesz representer is an inverse
probability weight or density ratio. By simplifying use of the EIF for
efficient estimation, and making its derivation more accessible in complex
settings, our work expands the set of tools for the construction of TMLE
algorithms and can be used to spur the application of targeted learning to new
applied science areas.

%%%%%%%%%%%%%%%%%%%%%%%%%%%%%%%%%%%%%%%%%%%%%%%%%%%%%%%%%%%%%%%%%%%%%%%%%%%%%%%
%\section*{Supplementary material}

%Supplementary material available online includes proofs of presented results.

\subsection*{Code Availability}

Code for reproducing the simulation studies and data analysis is available at
\url{https://github.com/nshlab/pub_code_RieszCML_ijb}. All analyses were
performed using version~4.5.1 of the \texttt{R} language and environment for
statistical computing~\citep{RCoreTeam}. The proposed estimators
are implemented in the open-source \texttt{R} package \texttt{\{RieszCML\}},
available at \url{https://github.com/nshlab/RieszCML}.

\subsection*{Acknowledgements}

SVB acknowledges support from the National Science Foundation (award no.~DGE
2140743). The authors thank Alex Luedtke for pointing out the prior work upon
which Algorithm~\ref{alg:seq-tmle}
is based, Rahul Singh for pointing out opportunities to clarify how our
contributions differ from those of \citet{Chernozhukov2022c}, and David Berger
for comments and corrections offered via an independent review of an earlier
draft. Anthropic's Claude Opus 5.5 was used to check and evaluate technical
claims during the revision process; however, all ideas and conclusions are the
authors' own, and the authors take full responsibility for the content of this
manuscript.

%%%%%%%%%%%%%%%%%%%%%%%%%%%%%%%%%%%%%%%%%%%%%%%%%%%%%%%%%%%%%%%%%%%%%%%%%%%%%%%
\bibliography{refs}
\end{document}